\documentclass[aps,pra,twocolumn,longbibliography,superscriptaddress]{revtex4-1}

\usepackage{amsmath}
\usepackage{amsthm}
\usepackage{amssymb}
\usepackage{mathtools} 
\usepackage{bm}
\usepackage{braket}
\usepackage{lineno}
\usepackage{pifont}
\usepackage{bbm}
\usepackage{bbding}
\usepackage{soul}

\usepackage[svgnames]{xcolor}
\usepackage[many]{tcolorbox}
\newtcolorbox{mybox}[1][]{
  drop shadow southeast,
  enhanced,colback=red!5!white,colframe=black!75!, width=0.36\textwidth,
  #1
}
\newtcolorbox{fancybox}[1][]{
  enhanced,drop fuzzy midday shadow,
  boxrule=1pt,arc=4pt,boxsep=0pt,
  left=.5em,right=.5em,top=1ex,bottom=1ex,
  colback=olive, width=0.46\textwidth,#1
}

\newtcbox{\testbox}[1][red]
  {on line, enhanced, drop shadow southeast,arc = 0pt, outer arc = 0pt,
    colback = #1!5!white, colframe = red!75!black,
    boxsep = 0pt, left = 1pt, boxrule = 0pt}
    
\newtheorem{theorem}{Theorem}

\newtheorem{lemma}{Lemma}

\usepackage{graphicx}
\usepackage{subfig}
\usepackage{booktabs}        

\usepackage[hidelinks]{hyperref}
\usepackage{url}
 
\usepackage{algorithm}
\usepackage{algpseudocode}
 
\usepackage{caption}
\usepackage{adjustbox}

\DeclareMathOperator{\Tr}{Tr}

\begin{document}
\title{Shuffle-QUDIO: accelerate distributed VQE with  trainability enhancement and measurement reduction}

\author{Yang Qian}
\thanks{This work was done when he was a research intern at JD Explore Academy}
\affiliation{School of Computer Science, Faculty of Engineering, The University of Sydney, Darlington, NSW 2008, Australia}
\author{Yuxuan Du}
\thanks{Corresponding author, duyuxuan123@gmail.com}
\affiliation{JD Explore Academy, Beijing 101111, China}
\author{Dacheng Tao}
\thanks{Corresponding author, dacheng.tao@gmail.com}
\affiliation{JD Explore Academy, Beijing 101111, China}
\affiliation{School of Computer Science, Faculty of Engineering, The University of Sydney, Darlington, NSW 2008, Australia}
\date{\today}

\begin{abstract}
    The variational quantum eigensolver (VQE) is a leading strategy that exploits noisy intermediate-scale quantum (NISQ) machines to tackle chemical problems outperforming classical approaches. To gain such computational advantages on large-scale problems, a feasible solution is the \textbf{QU}antum \textbf{DI}stributed \textbf{O}ptimization (QUDIO) scheme, which partitions the original problem into $K$ subproblems and allocates them to $K$ quantum  machines followed by the parallel optimization. Despite the provable acceleration ratio, the efficiency of QUDIO may heavily degrade  by  the synchronization operation. To conquer this issue, here we propose Shuffle-QUDIO to involve shuffle operations into local Hamiltonians during the quantum distributed optimization. Compared with QUDIO, Shuffle-QUDIO significantly reduces the communication frequency among quantum processors and simultaneously achieves better trainability. Particularly, we prove that Shuffle-QUDIO enables a faster convergence rate over QUDIO. Extensive numerical experiments are conducted to verify that Shuffle-QUDIO allows both a wall-clock time speedup and low approximation error in the tasks of estimating the ground state energy of molecule. We empirically demonstrate that our proposal can be seamlessly integrated with other acceleration techniques, such as operator grouping, to further improve the efficacy of  VQE. 
\end{abstract}

\maketitle

\section{Introduction}\label{sec:intro}

Quantum computing is expected to demonstrate advantages over classical computers in dealing with certain tasks, such as boson sampling \cite{spring2013boson,wang2017high,bulmer2021boundary} and integer factorization \cite{jiang2018quantum,peng2019factoring}. With the advent of noisy intermediate-scale quantum (NISQ) era \cite{preskill2018quantum,bharti2021noisy}, Google has experimentally verified that  when sampling the output of a pseudo-random quantum circuit, current NISQ devices can run   faster than the state-of-the-art classical computers \cite{arute2019quantum}. Recently, USTC has achieved more difficult sampling tasks on \textit{Zuchongzhi} 2.1 to further push the frontier of quantum computational advantages \cite{zhu2021quantum}. However, the unavoidable system noise and the restricted coherence time prevent the execution of complicated quantum algorithms on NISQ devices. To accommodate the limitations of NISQ machines, variational quantum algorithms (VQAs) \cite{mcclean2016theory,cerezo2021variational,cerezo2020variational,qian2021dilemma,tian2022recent} which employ a classical optimizer to train a parametrized quantum circuit, have emerged. Concisely, VQAs alternately interact between quantum circuits and classical optimizers, while the former evolves the quantum state and outputs classical information by measurements, and the latter is responsible for seeking the best parameters of quantum circuit to minimize the discrepancy between the predictions and the targets. Pioneer studies have verified the power of VQAs in quantum finance \cite{orus2019quantum,pistoia2021quantum}, quantum chemistry \cite{grimsley2019adaptive,arute2020hartree,kandala2017hardware,robert2021resource,kais2014introduction,wecker2015solving,cai2020quantum,wang2019accelerated,romero2018strategies,Cervera2021meta-variational,parrish2019quantum}, many-body physics \cite{huang2021provably,lee2021neural,endo2020variational}, machine learning \cite{huang2021power,huang2022quantum,du2021exploring,caro2022generalization,gili2022quantum}, and combinational optimization \cite{farhi2014quantum,zhou2020quantum,harrigan2021quantum,lacroix2020improving,hadfield2019quantum,zhou2022qaoa} from both theoretical and experimental aspects.
 
Although VQAs promise the practical  applications of NISQ machines, they are challenged by the scalability issue. The required number of measurements for VQAs scales with $O(poly(n, 1/\epsilon))$ with $n$ being the problem size and $\epsilon$ being the tolerable error, which implies an expensive runtime for large-scale problems. One canonical instance is the variational quantum eigensolver (VQE) \cite{peruzzo2014variational}, which is developed to estimate the low-lying energies and corresponding eigenstates of molecule systems.  VQE  contains two key steps. First, the electronic Hamiltonian is reformulated to the qubit Hamiltonian $H=\sum_{i=1}^M\alpha_iH_i$ through Jordan-Wigner, Bravyi-Kitaev, or parity transformations \cite{seeley2012bravyi,bravyi2002fermionic,jordan1993paulische}, where $H_i\in\{\sigma_X, \sigma_Y, \sigma_Z, \sigma_I\}^{\otimes n}$ and $\alpha_i\in \mathbb{R}$ for $\forall i \in [M]$, and $M$ is the number of Pauli operators. The property of $H$ is then estimated by a variational quantum circuit whose parameters are updated by a classical optimizer. Principally, it requires $O(poly(M, 1/\epsilon))$ queries to the quantum circuit   in each iteration to collect the updating information \cite{gonthier2020identifying}. With this regard, VQEs towards large-scale molecules request an intractable time expense on the measurements. This scalability issue impedes the journey of VQEs to the quantum advantages.

Approaches for reducing the computational overhead of quantum measurements in VQE can be roughly classified into five categories, including operator grouping \cite{ralliImplementationMeasurementReduction2020,verteletskyiMeasurementOptimizationVariational2020,zhao2020measurement,gokhaleMinimizingStatePreparations2019}, ansatz adjustment \cite{tkachenko2021correlation,zhang2022variational}, shot allocation \cite{arrasmith2020operator,van2021measurement,gu2021adaptive}, classical shadows \cite{huang2020predicting,hadfield2022measurements}, and distributed optimization \cite{andres2019automated,barratt2003parallel,duAcceleratingVariationalQuantum2021,mineh2022accelerating}. Specifically, the operator grouping strategy focuses on finding the commutativity between local Hamiltonian terms $\{H_i\}$ in $H$. The commutable Hamiltonians can be evaluated by the same measurements, which enable the measurement reduction \cite{kandala2017hardware,ralliImplementationMeasurementReduction2020,verteletskyiMeasurementOptimizationVariational2020,zhao2020measurement,gokhaleMinimizingStatePreparations2019}. Ansatz adjustment targets to tailor the layout of ansatz to reduce the circuit depth \cite{tkachenko2021correlation,tang2021qubit,grimsley2019adaptive} or the number of qubits \cite{zhang2022variational}. For example, Ref.~\cite{tkachenko2021correlation} attempts to assign two qubits with stronger mutual information to the adjacent locations with direct connectivity on the physical quantum chips, leading to shallower circuits over the original VQEs to reach a desired accuracy. Shot allocation aims to assign the number of shots among $\{H_i\}$ in a more intelligent way. A typical solution is to allocate more shots to the terms with a larger coefficient $|\alpha_i|$ and a larger variance of $\braket{H_i}$. Another measurement reduction method, classical shadows, constructs an approximate classical representation of a quantum state based on few measurements of the state \cite{huang2020predicting}. With this representation, $O(\log(M))$ measurements are enough to estimate the expectation value of whole observable with high precision.

On par with engineering the quantum part, we can accelerate the optimization of VQE by using multiple quantum processors (workers), inspired by the success of distributed optimization in deep learning and the growing number of available quantum chips. There are generally two types of distributed VQAs. The first paradigm is decomposing the primal quantum systems into multiple smaller circuits and running them in parallel  \cite{barratt2020parallel,diadamo2021distributed}. The second paradigm is  utilizing the quantum cloud server in which the problem Hamiltonian can be pre-divided into several partitions and distributed into $Q$ local quantum workers respectively. Each worker estimates the expectation value of partial local Hamiltonians with no more than $O(poly(M/Q))$ queries and delivers the result to the rest workers after a single iteration. Noticeably, such a methodology inevitably encounters the communication bottleneck, quantum circuit noise, and the risk of privacy leakage. As such,   Ref.~\cite{duAcceleratingVariationalQuantum2021} devises the \textbf{QU}antum   \textbf{DI}stributed \textbf{O}ptimization (QUDIO), a novel distributed-VQA scheme in a lazy communication manner, to address this issue. Unfortunately, the naive allocation method is not suitable for VQEs since the coefficients $\{\alpha_i\}$ of the local Hamiltonian terms $\{H_i\}$ are varied, leading to unbalanced contributions to the overall variance of Hamiltonian estimation. Such an estimation error can be exacerbated by the increased communication interval,  which renders the trade-off between the acceleration ratio and the approximation error of VQEs.

To maximally suppress the negative effects of large communication interval on the convergence rate, here we propose a new quantum distributed optimization framework, called Shuffle-QUDIO. Different from QUDIO, for every local worker, the local Hamiltonian terms are randomly shuffled and sampled without replacement according to the worker's rank before each iteration. From the statistically view, this operation alleviates the issue such that every local worker may only observe incomplete local Hamiltonians during the optimization. Moreover, the dynamic allocation of Hamiltonian terms alleviates the accumulated deviation with respect to the target Hamiltonian $H$ after a large number of local updates. In this way, Shuffle-QUDIO achieves performance improvements while keeping low communication cost. Another advantage of our proposal is its compatibility with all types of quantum hardware. This assures its potential of unifying existing quantum devices to accelerate the training of VQEs.

To theoretically exhibit the advance of our proposal, we prove that Shuffle-QUDIO allows a faster convergence rate than that of QUDIO. By leveraging the non-convex optimization theory, we exhibit that the dominate factors effecting the  convergence rate are the number of distributed quantum machines $K$, the local updates (communication interval) $W$, and the global iterations $T$, i.e., $O(poly(W, K, 1/T))$. To benchmark the performance of Shuffle-QUDIO, we conduct systematic numerical experiments on VQEs under both fault-tolerant and noisy scenarios. The achieved results confirm that Shuffle-QUDIO achieves smaller approximation error over QUDIO, as well as lower communication overhead among clients and server, and sub-linear speedup ratio. In addition, we demonstrate that the performance of Shuffle-QUDIO under the noisy setting can be further boosted by combining the advanced operator grouping strategy.

The remaining parts of this paper are organized as follows. Section \ref{sec:bg} briefly introduces the preliminary knowledge about the optimization of variational quantum circuits. Section \ref{sec:sq} presents the pipeline of the proposed algorithm and presents the convergence analysis. Section \ref{sec:nr} exhibits numerical simulation results. Section \ref{sec:dis} gives a summary and discusses the outlook.

\section{Preliminary}\label{sec:bg}

The essence of VQE is tuning  an $n$-qubit parameterized quantum state $\rho({\bm{\theta}})=\ket{\psi(\bm{\theta})}\bra{\psi(\bm{\theta})}$ with $\bm{\theta}\in\mathbb{R}^P$ to minimize the energy of a problem Hamiltonian 
\begin{equation}\label{eqn:ham-def}
	H=\sum_{i=1}^M\alpha_iH_i\in \mathbb{C}^{2^n\times 2^n},
\end{equation}
where  $H_i$  refers to the  $i$-th local Hamiltonian term with the weight $\alpha_i$. The energy minimization is formulated by the  loss function 
\[L(\bm\theta,H):=\Tr(\rho({\bm{\theta}})H)=\sum_{i=1}^M\alpha_i\Tr(\rho({\bm{\theta}})H_i).\] With a slight abuse of notation, we denote $H_i$ as $\alpha_iH_i$ and simplify the above loss function as \[L(\bm\theta,H)=\sum_{i=1}^M\Tr(\rho({\bm{\theta}})H_i).\]
The parameterized quantum state is  prepared by an ansatz with $\ket{\psi(\bm{\theta})} =U(\bm{\theta})\ket{\phi}$ and $\ket{\phi}$ being an initial quantum state. A generic form of $U(\bm{\theta})$ is  
\begin{equation}\label{eqn:circuit}
    U(\bm{\theta}) = \prod_{l=1}^LU_{e}\prod_{i=1}^N\exp(-i\theta_{l,i}O_i),
\end{equation}
where $O_i$ is a Hermitian matrix and $U_{e}$ denotes a fixed unitary  composed of multi-qubit gates. By iteratively updating the circuit parameters $\bm{\theta}$ to minimize the loss, the quantum state $\rho({\bm{\theta}})$ is expected to approach the eigenstate of $H$ with the minimum eigenvalue.

\subsection{Optimization of VQE}

Gradient descent (GD) based optimizers are widely used in previous literatures of VQE. The parameters $\bm{\theta}^{t+1}$ at the $(t+1)$-th iteration is updated alongside the steepest descent direction with learning rate $\eta$, i.e.,
\begin{equation}\label{eq:gd}
    \bm{\theta}^{t+1}=\bm{\theta}^t-\eta\nabla L(\bm{\theta}^t,H).
\end{equation}
Unlike classical neural networks that utilize gradient back-propagation  to update parameters \cite{lecun1988theoretical}, VQE adopts the parameter-shift rule \cite{banchi2021measuring,wierichs2022general} to obtain the unbiased estimation of the gradient. The gradient with respect to the $i$-th parameter is  
\begin{equation}
    \frac{\partial L(\bm{\theta},H)}{\partial\theta_i}=\frac{L(\bm{\theta}+\frac{\pi}{2}\bm{e}_i,H)-L(\bm{\theta}-\frac{\pi}{2}\bm{e}_i,H)}{2},
\end{equation}
where $\bm{e}_i$ denotes the indicator vector for the $i$-th element of parameter vector $\bm{\theta}$. When the number of trainable parameters is  $P$, the required number of measurements to complete the gradient computation scales with $O(poly(PM))$ without applying any measurement reduction strategies.  

\subsection{Optimization of the distributed VQE}

To accelerate the training of VQA, Ref.~\cite{duAcceleratingVariationalQuantum2021} proposed the \textbf{QU}antum \textbf{DI}stributed \textbf{O}ptimization (QUDIO) scheme. The key idea of QUDIO is to partition the problem Hamiltonian $H$ in Eq.~(\ref{eqn:ham-def}) into several groups and distribute them into multiple quantum processors to be manipulated in parallel.  Mathematically, suppose that there are $K$ available quantum processors $\{\mathcal{Q}_i\}_{i=1}^K$, the Hamiltonian terms $\{H_i\}_{i=1}^M$ are divided into $K$ subgroups $\{\mathcal{S}_i\}_{i=1}^K$, where $\mathcal{S}_i=\cup_{j\in S_i} \{H_j\}$, so that $\sum_{i=1}^K |S_i| =M$ and $S_i\cap S_j=\emptyset$ when $i\neq j$. 

In the initialization process, the $i$-th subgroup $\mathcal{S}_i$ is allocated to the $i$-th quantum processor $\mathcal{Q}_i$ for $\forall i \in [K]$. All local processors share the same initial parameters $\bm{\theta}^0$   with $\bm{\theta}_i^{(0,0)}=\bm{\theta}^0$ for $\forall i\in [K]$. The subsequent training process alternately switches between the local updates and the global synchronization. During the phase of local updates, each quantum processor follows the gradient descent rule to update the parameters to minimize the local loss function $L(\bm{\theta}_i,H_{S_i})=\sum_{j\in S_i}\Tr(\rho({\bm{\theta}_i})H_j)$, i.e., the parameters of the $i$-th processor  at the $(t,w)$-th step is updated as Eq.~(\ref{eq:gd}). After fulfilling $W$ local updates, all parameters from distributed quantum processors are synchronized by averaging the collected parameters $\bm{\theta}^{t+1}=\frac{1}{K}\sum_{i=1}^K\bm{\theta}_i^{(t,W)}$. Repeating the above two phases until the termination conditions (e.g. the maximum number of iterations) are met,    the synchronized parameters are returned as the final parameters.

Ignoring the communication overhead among quantum processors, QUDIO with $W=1$ is expected to linearly accelerate the optimization  of VQE. However, the communication bottleneck could degrade the acceleration efficiency. An optional solution is to increase $W$ to reduce the communication frequency. As indicated in \cite{duAcceleratingVariationalQuantum2021}, the performance of VQA witnesses a rapid drop with the increased $W$.

\section{Shuffle-QUDIO for  VQE}\label{sec:sq}

The performance of  QUDIO suffers from a high sensitivity of the communication interval. Intuitively, this issue originates from the fact that each quantum processor in QUDIO only perceives a static subset of the whole observable set during the entire training process. The $i$-th processor updates its local parameters based on the partial observations before communicating with other processors. Meantime, the coefficients $\{\alpha_i\}$ and the variance of Pauli operators $\{H_i\}$ differ from each other, leading to different contributions to the expectation estimation of $H$. As a result, the local processor fails to characterize the full property of the problem Hamiltonian $H$. With  multiple local updates, the accumulation of bias further degrades the performance of QUDIO. To tackle this issue, here we devise a novel  quantum distributed optimization scheme, called Shuffle-QUDIO, to avoid the performance drop when synchronizing in a low frequency. 

\begin{figure*}[htp]
\captionsetup[subfigure]{justification=centering}
\centering
\includegraphics[width=0.9\textwidth]{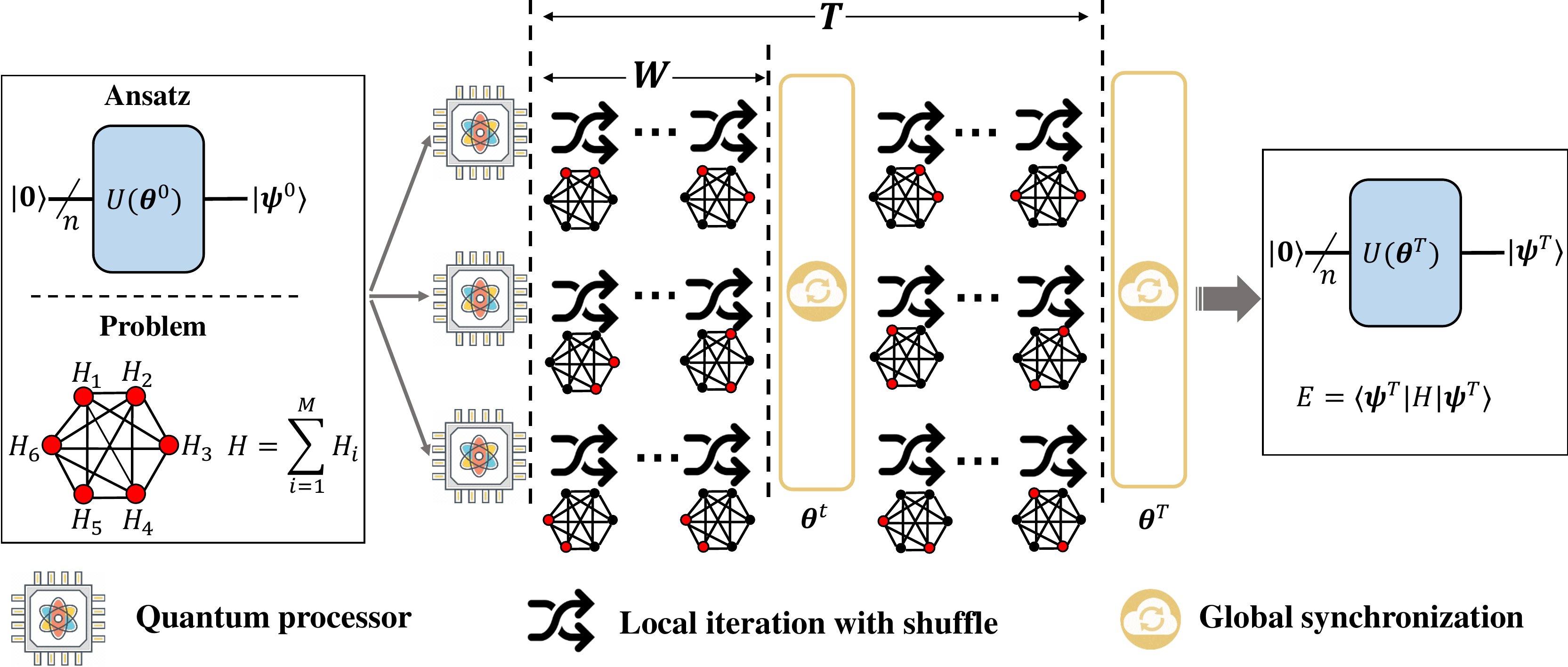}
\caption{\small{\textbf{The scheme of Shuffle-QUDIO.} The Shuffle-QUDIO consists of three subroutines, including initialization, local updates and global synchronization. During the phase of initialization, multiple copies of the original ansatz and the corresponding problem Hamiltonian $H$ are dispatched into each local processor. Note that each processor shares the same seed of the random number generator. For each iteration in the local updates, the set of observables $\{H_i\}_{i=1}^M$ is \textit{randomly shuffled} and the $i$-th local processor picks the subset of whole observables according to the assigned random number. In this way, the observables of each processor do not overlap with each other and the union of their observables exactly constitutes the problem Hamiltonian $H$. After $W$ local updates, the parameters of each local ansatz are aggregated and then reassigned to all local processors, which is called global synchronization. When the maximal number $T$ of iterations is reached, Shuffle-QUDIO executes the final synchronization and outputs the trained parameters.}}
\label{fig:scheme}
\end{figure*}

\subsection{Algorithm descriptions}

The paradigm of Shuffle-QUDIO is  depicted in Fig.~\ref{fig:scheme}, which consists of three steps.
\begin{enumerate}
    \item \textbf{Initialization}. The variational quantum circuit  $U(\bm{\theta})$ in Eq.~(\ref{eqn:circuit})  of each quantum processor is initialized with the same parameters $\bm{\theta}_i^{(0,0)}=\bm{\theta}^0$ for $i=\{1,...,K\}$ and all local Hamiltonian terms $\{H_i\}$ are distributed to each processor.
    \item \textbf{Local updates}. Each processor independently updates the parameters $\bm{\theta}^{(t,w)}_i$ following the gradient descent principle. First, Shuffle-QUDIO randomly shuffles the sequence of local Hamiltonian terms. Note that the random number of each processor is generated from the same random seed. Assuming the permutation vector is denoted by $\pi^{(t,w)}$, the visible Hamiltonians for the $i$-th processor at the $t$-th iteration are $\mathcal{H}^{(t,w)}_i=\{H_{\pi^{(t,w)}(j)}\}_{j=\frac{M}{K}(i-1)+1}^{\frac{M}{K}i}$ (suppose $M$ is exactly divided by $K$). Then each processor estimates the gradient $\bm{g}_i^{(t,w)}$ by the parameter-shift rule. Note that $\bm{g}$ denotes the estimated gradient on the quantum device due to the finite number of measurements, while $\nabla L$ refers to the corresponding accurate gradient. The parameters are updated as
    \begin{equation}
        \bm{\theta}^{(t,w+1)}_i=\bm{\theta}^{(t,w)}_i-\eta \bm{g}_i^{(t,w)},
    \end{equation}
    where $\eta$ is the learning rate.
    Repeat the above local updates for $W$ local steps.
    \item \textbf{Global synchronization}. Once the local updates are completed, the central server synchronizes parameters among all quantum processors in an averaged manner, i.e.,
    \begin{equation}
        \bm{\theta}^{t+1}=\frac{1}{K}\sum_{i=1}^K \bm{\theta}^{(t,W)}_i.
    \end{equation}
    If the number of the global iterations reaches $T$,  the parameters $\bm{\theta}^T$ are returned as the output; otherwise, return back to step 2.
\end{enumerate}
The pseudo code of Shuffle-QUDIO is summarized in Fig.~\ref{alg:shuffle-qudio}. Compared with conventional VQE which sequentially measures the expectation value of every single observable, the strategy of distributed parallel optimization accelerates the estimation of the complete observables by $K$ times. Furthermore, the shuffling operation alleviates the deviation of the optimization direction during the local updates and thus warrants a stabler performance after increasing communication interval. This is because in a statistical view, each processor can leverage the  information of all local Hamiltonian terms to update local parameters in the training process. 
\begin{lemma}\label{the:shuffle}
Let $\{H_1,...,H_M\}$ be $M$ Hermitian matrices in $\mathbb{C}^{2^n\times 2^n}$, $H=\sum_{i=1}^MH_i$. Let $\rho(\bm{\theta})$ be an $n$-qubit quantum state parameterized by $\bm{\theta}$. For any $k\in \{1,...,M\}$, let $H_{\pi(1)},...,H_{\pi(k)}$ be uniformly sampled without replacement from $\{H_1,...,H_M\}$. Let $L=\Tr(\rho(\bm{\theta})H)$ and $L_m=\Tr(\rho(\bm{\theta})\sum_{i=1}^mH_{\pi(i)})$. Then we have
\begin{equation}
    \mathbb{E}[\frac{\partial L_m}{\partial \bm{\theta}}]=\frac{m}{M}\frac{\partial L}{\partial \bm{\theta}}.
\end{equation}
\end{lemma}
\noindent Refer to Appendix \ref{pr:shuffle} for proof details. Lemma \ref{the:shuffle} implies that the direction of the expected gradient of each local quantum processor in Shuffle-QUDIO is unbiased. This guarantees that the local quantum circuits are individually optimized forward along the right direction when they do not communicate frequently with each other, which narrows the performance gap between a single processor and the synchronized model. By contrast, during the local updates of QUDIO, there always exists a bias between the locally estimated gradient and the global gradient. Specifically, Shuffle-QUDIO achieves smaller gradient deviation than vanilla QUDIO, as indicated by the following lemma, whose proof is provided in Appendix \ref{pr:grad-bias}.
\begin{lemma}\label{lemma:grad-bias}
Assume the norm of local gradient $\bm{g}_k(\bm{\theta}, H_k)$ is bounded by $||\bm{g}_k||^2\leq G^2$. Compared with QUDIO, the discrepancy between the local gradient $\bm{g}_k(\bm{\theta}, H_k)$  and the global gradient $||\nabla L(\bm{\theta},H)||$ in Shuffle-QUDIO is reduced from $2(K^2+1)G^2$ to $(K-1)^2G^2$.
\end{lemma}

\begin{figure}[htp]
\begin{algorithm}[H]
   \begin{algorithmic}[1]
   \State \textbf{Input}: The initialized parameters $\bm{\theta}^{(0)}\in [0,2\pi)^{P}$, the employed loss function $\mathcal{L}=\Tr(\rho(\bm{\theta})H)$, the given Hamilton $H=\sum_{i=1}^{M}H_i\in \mathbb{C}^{2^n\times 2^n}$, the hyper-parameters $\{K,\eta, W,T\}$  \;
   \State Initialize the permutation vector $\pi^{(0)}=[1,2,...,M]$ of Hamiltonian term $H_i$ in order \;
 \For{$t=0,\cdots,T-1$} \;
\For{Quantum processor $\mathcal{Q}_i$, $\forall i\in[K]$ \textbf{in parallel}} \;
  \State $\bm{\theta}_{i}^{(t,0)}= \bm{\theta}^{(t)}$ \;
  \State $\pi^{(t,0)}=\pi^{(t)}$
 \For{$w=1, \cdots, W-1$}
  \State Obtain $\pi^{(t,w)}$ by randomly shuffling the elements of $\pi^{(t,w-1)}$ \;
  \State Obtain the subset of Hamiltonian terms $\mathcal{H}^{(t,w)}_i=\{H_{\pi^{(t,w)}(j)}\}_{j=\frac{M}{K}(i-1)+1}^{\frac{M}{K}i}$ according to processor's rank $i$ \;
  \State Compute the estimated gradients $g_{i}^{(t,w)}$ \;
  \State Update  $\bm{\theta}_{i}^{(t, w + 1)}=\bm{\theta}_{i}^{(t,w)} - \eta g_{i}^{(t,w)}$ \;
 \EndFor
\EndFor
  \State Synchronize $\bm{\theta}^{(t + 1)} =   \frac{1}{K}\sum_{i=1}^K \bm{\theta}_{i}^{(t,W)}$
\EndFor
\State \textbf{Output:} $\bm{\theta}^{T}$
\end{algorithmic}
\end{algorithm}
\caption{\small{\textbf{The Pseudocode of Shuffle-QUDIO}.}}
\label{alg:shuffle-qudio}
\end{figure}

\subsection{Convergence analysis}
We next show the convergence guarantee of Shuffle-QUDIO. When running VQE on NISQ devices, the system imperfection introduces noise into the optimization. To this end, we consider the worst scenario in the convergence analysis, where the system noise is modeled by the depolarizing channel. Mathematically, the depolarizing channel $\mathcal{N}_p$ transforms the quantum state  $\rho\in\mathbb{C}^{2^n\times 2^n}$ to $\mathcal{N}_p(\rho)=(1-p)\rho+p\mathbb{I}/2^n$, and with increasing the noise strength $p$, the quantum state finally evolves to the maximally mixed state. As proved in \cite{du2020learnability}, the depolarizing channel applied on each circuit depth can be merged at the end of the quantum circuit. Therefore, without loss of generality, the estimated gradient with respect to the $i$-th parameter is
\begin{equation}
    \frac{\partial \overline{L}(\bm{\theta},H)}{\partial \theta_i}=(1-p)\frac{\partial L(\bm{\theta},H)}{\partial \theta_i}.
\end{equation}
The convergence rate of Shuffle-QUDIO is summarized by the following theorem whose proof is provided in Appendix \ref{app:proof1}.
 
\begin{theorem}\label{the:convergence}
Let the gradient of loss function $L$ be $F$-Lipschitz continuous, $G$ be the upper bound of the gradient norm, $\eta$ be the learning rate of optimizer, $p$ be the strength of depolarizing noise, $K$ and $W$ be the number of distributed quantum processor and local iterations respectively, the convergence of Shuffle-QUDIO in the noisy scenario is summarized as
\begingroup
\allowdisplaybreaks
    \begin{align}
        &\frac{1}{T}\sum_{t=1}^T||\nabla L(\bm\theta^t)||^2\leq \frac{2(L(\bm\theta^1)-L(\bm\theta^{T+1}))}{\eta T}\nonumber\\
        &+\frac{4F^2\eta^2W^2G^2(K-1)}{KT}\nonumber\\
        &+\frac{(2K(K-2+2p)+(\eta F+1)(1-p)^2)G^2}{T}
    \end{align}
\endgroup
\end{theorem}
\noindent Theorem \ref{the:convergence} reveals that an increased quantum noise rate $p$ leads to poor convergence of Shuffle-QUDIO, which emphasizes the significance of error mitigation \cite{endo2018practical,endo2021hybrid,strikis2021learning,du2020quantumsearch} in quantum optimization. Meantime, the shorter communication interval $W$ among distributed quantum processors guarantees a better performance of Shuffle-QUDIO. Note that although a large $W$ still hinders the distributed optimization, Shuffle-QUDIO achieves a relatively smaller sensitivity to $W$ than that of QUDIO.

From the technical view, although the proof of Theorem \ref{the:convergence} is derived from the classical results on local SGD \cite{haddadpour2019local}, there are some key differences between them. First, in classical local SGD, each worker independently samples a mini-batch from the whole dataset without other limitations. By contrast, the distributed quantum processors randomly sample the local Hamiltonian terms without replacement in each local iteration, which means that the Hamiltonian terms of each processor do not overlap and the union exactly constitutes the complete molecule Hamiltonians. This special sampling method guarantees the integrity of the problem Hamiltonian, but poses a challenge for theoretical analysis. Second, our analysis does not rely on the strong assumptions, such as convexity or Polyak-Lojasiewicz (PL) condition \cite{sweke2020stochastic}. Furthermore, the quantum noise in NISQ devices inevitably shifts the quantum state and biases the estimated gradients, which differentiates VQE from classical machine learning.

\section{Numerical results}\label{sec:nr}
To verify the effectiveness of Shuffle-QUDIO, we apply it to estimate the ground state of several molecules with the lowest energy. Jordan-Wigner transformation \cite{jordan1993paulische} is employed to transform these electronic Hamiltonians into the qubit Hamiltonians represented by Pauli operators. For example, the LiH system is totally described by $12$ qubits and $631$ local Pauli terms $\{\sigma_I,\sigma_X,\sigma_Y,\sigma_Z\}^{\otimes 12}$. The ansatz is designed in a hardware-efficient style inspired by \cite{kandala2017hardware}, whose layout is shown in Fig.~\ref{fig:hea}. We conduct numerical experiments on classical device with Intel(R) Xeon(R) Gold 6267C CPU @ 2.60GHz and 128 GB memory. For each setting, the experiment is repeated for $5$ times with different random seeds to mitigate the effect of randomness. Stochastic gradient descent is used to update trainable parameters, where the learning rate is set as $\eta = 0.4$.

\begin{figure}[htp]
\captionsetup[subfigure]{justification=centering}
\centering
\includegraphics[width=0.4\textwidth]{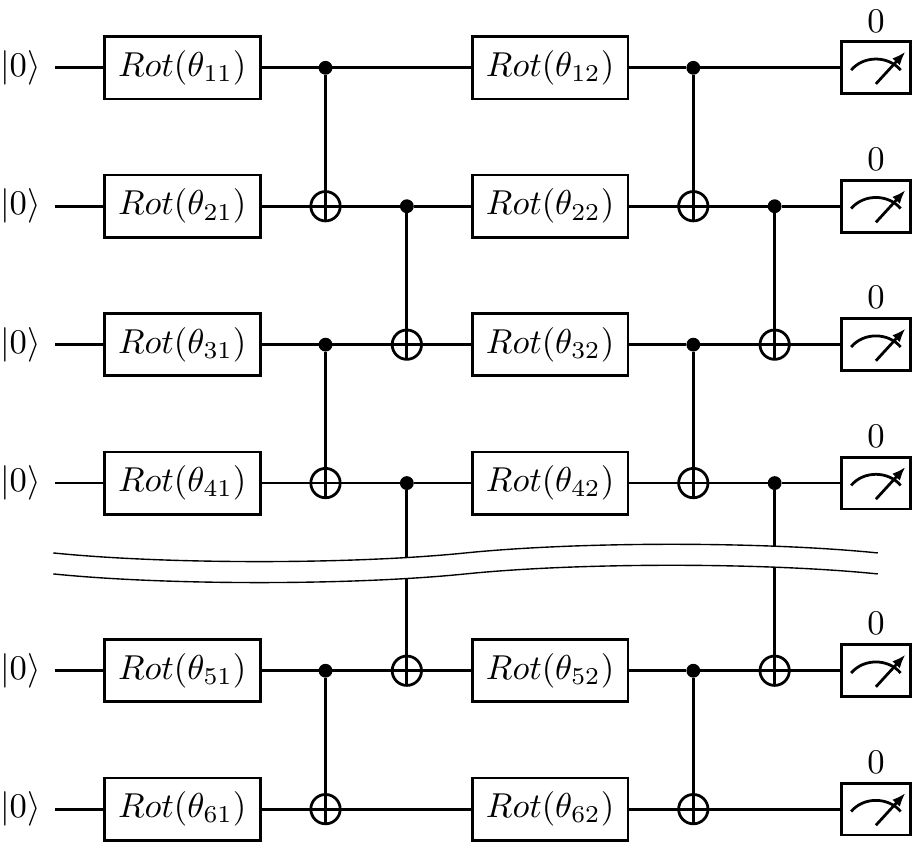}
\caption{\small{\textbf{Layout of hardware-efficient ansatz.} The gate `\textit{Rot}' represents the concatenation of rotation gate $R_z$, $R_y$, $R_z$, and $\theta$ represents the rotation angle.}}
\label{fig:hea}
\end{figure}

\begin{figure}[htp]
\centering
\includegraphics[width=0.48\textwidth]{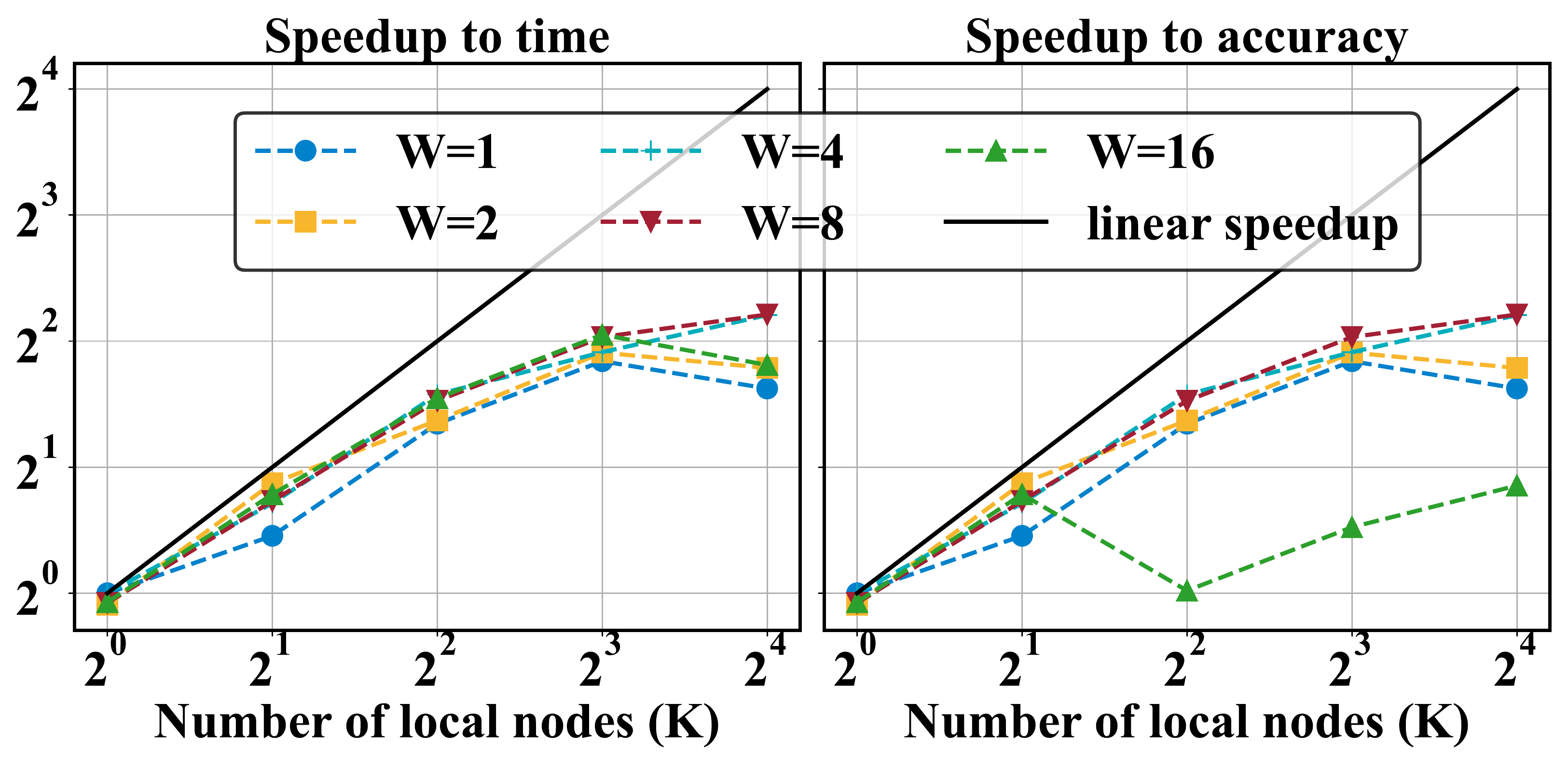}
\caption{\small{\textbf{Speedup of Shuffle-QUDIO to VQE for LiH.} The label `$W=a$' refers that the number of local iterations is $a$. The label `\textit{linear speedup}' represents the reference line of the linear speedup.}}
\label{fig:speedup2}
\end{figure}

\begin{figure*}[htp]
\captionsetup[subfigure]{justification=centering}
\centering
\includegraphics[width=0.9\textwidth]{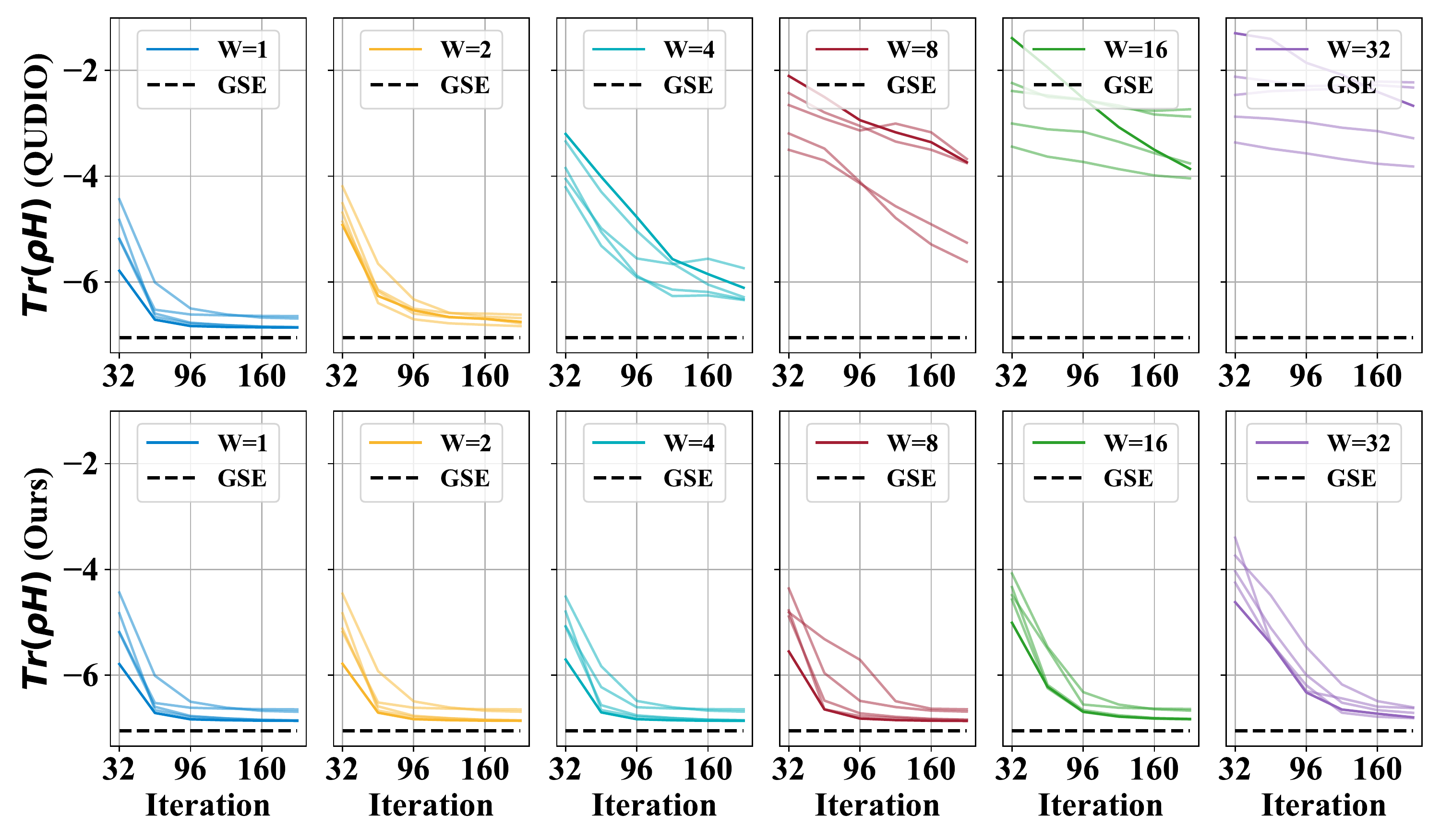}
\caption{\small{\textbf{Training process of VQE optimized by QUDIO and Shuffle-QUDIO respectively.} Each data point is collected after the synchronization. The dashed black line denotes the exact ground state energy (GSE) at the same setting.   \textit{The first row:} the loss curve with respect to iterations in QUDIO. With exponentially increasing $W$, the convergence of training is severely degraded, as depicted in subplot at first row, sixth column. \textit{The second row:} the loss curve with respect to the iterations in Shuffle-QUDIO. The speed of loss decrease sees a relatively slow decay with $W$ growing. When $W=32$ (second row, sixth column), the loss still converges to the same level of $W=1$ within $200$ iterations.}}
\label{fig:loss}
\end{figure*}

\subsection{Acceleration ratio}

We first explore the speedup of Shuffle-QUDIO. Specifically, under the setting of $K$ quantum processors and $W$ local iterations, we record the time spent per iteration as $t_1^{K,W}$ and the time spent training the model to reach a specific precision as $t_2^{K,W}$.  The speedup to time $s_1^{K,W}$ and the speedup to accuracy $s_2^{K,W}$ are defined as $s_1^{K,W}={t_1^{K,W}}/{t_1^{1,1}}$ and $s_2^{K,W}={t_2^{K,W}}/{t_2^{1,1}}$, respectively. Fig.~\ref{fig:speedup2} demonstrates the results when solving the ground state of LiH. As shown in the left panel of Fig.~\ref{fig:speedup2}, with growing the number of the distributed quantum processors, the metric $s_1$ witnesses a linear growth (from $K=1$ to $K=4$, $s_1$ reaches $3$ from $1$) and then gradually trends gently (from $K=4$ to $K=16$, $s_1$ reaches around $4$ from $3$) because of the communication bottleneck for a large number of nodes. To alleviate the communication bottleneck, we can increase $W$ to reduce the communication frequency and hence further improve the metric $s_1$. However, an overwhelmingly larger $W$ may lead to a poor convergence and then deteriorate the speedup to accuracy $s_2$. As indicated by the right panel of Fig.~\ref{fig:speedup2}, when $W\geq 16$, the speedup to accuracy $s_2$ suffers from a rapid drop, which results from the worse convergence brought by a small number of global synchronization.

\subsection{Sensitivity to communication frequency}

We compare QUDIO with Shuffle-QUDIO to show how the increased number of local iterations $W$ effects their performance under the ideal  scenario. The molecules LiH with varied inter-atomic length, i.e, $0.3\mathrm{\AA}$ to $1.9\mathrm{\AA}$ with step size $0.2\mathrm{\AA}$, are explored.  For QUDIO, the entire set of Pauli terms constituting the problem Hamiltonian is uniformly partitioned into $32$ subsets and distributed into $32$ local quantum processors. The accessible  Hamiltonian terms for each local processor remain fixed during the whole training process. The number of local iterations $W$ varies in $\{1,2,4,8,16,32\}$.

The simulation results of VQE for the molecule LiH with $0.5\mathrm{\AA}$   are illustrated in Fig.~\ref{fig:loss}. Because  the number of local iterations $W$ varies among different settings, we uniformly collect data point after every $32$ iterations (i.e., the least common multiple of all $W$) to guarantee the loss is obtained exactly after synchronization. The first row of Fig.~\ref{fig:loss} records the loss curves of QUDIO  with respect to the training steps under different local iterations $W$.  QUDIO experiences a severe drop of performance, and an evident gap between the estimated and the exact results appears when $W\geq 8$. By contrast, as shown in the bottom row of Fig.~\ref{fig:loss}, Shuffle-QUDIO well estimates the exact ground energy even when $W=32$. Comparing the subplots of the same column, Shuffle-QUDIO shows a distinct advantage in improving the convergence of the distributed VQE when requiring a lower communication overhead. For example, Shuffle-QUDIO achieves $-6.8Ha$ at the $96$-th iteration with $W=8$, while QUDIO only reaches $-4.1Ha$.

\begin{figure*}[htp]
\captionsetup[subfigure]{justification=centering}
\centering
\includegraphics[width=0.9\textwidth]{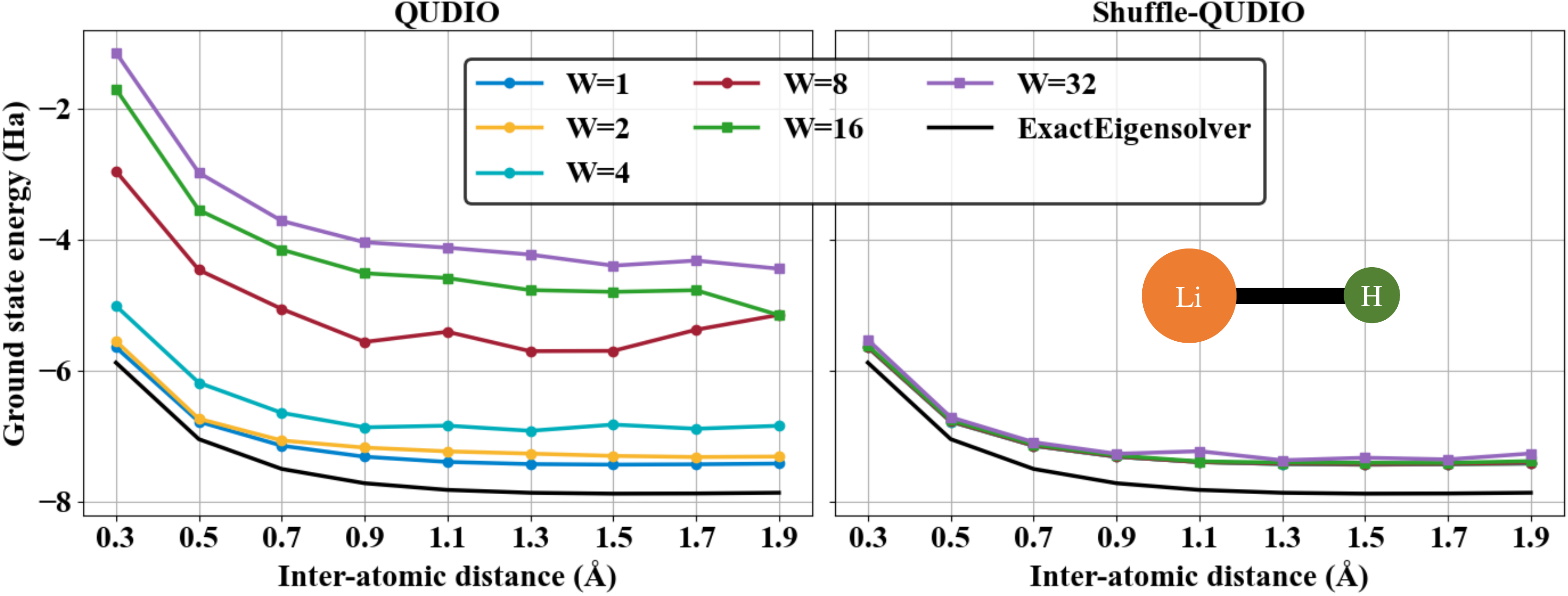}
\caption{\small{\textbf{Energy potential surface of molecule LiH.} The black line with label `\textit{ExactEigensolver}' represents the exact energy potential surface of the molecule LiH.}}
\label{fig:energy_surface}
\end{figure*}

The potential energy surface of LiH solved by  the conventional VQE is shown in Fig.~\ref{fig:energy_surface}, where the left panel describes the results of QUDIO with the varied number of local iterations $W$ and the right panel records the results of Shuffle-QUDIO. There exists a distinct boundary among the potential energy surfaces estimated by the different level of $W$ in QUDIO. More precisely, the estimated potential energy surface is gradually away from the exact potential energy surface (black line) with the increased $W$, which reveals the vulnerability of QUDIO when reducing the communication frequency among distributed workers. By contrast, Shuffle-QUDIO exhibits a fairly stable performance even when increasing $W$ from $1$ to $32$, drawn from the nearly coincident curves of potential energy surface at each setting of $W$. Note that the slight gap between the exact potential energy surface and the optimal estimated results originates from the restricted expressive power of the employed ansatz, which does not guarantee the prepared state definitely covers the ground state of LiH.

\begin{figure}[htp]
\captionsetup[subfigure]{justification=centering}
\centering
\subfloat[]{
\centering
\includegraphics[width=0.45\textwidth]{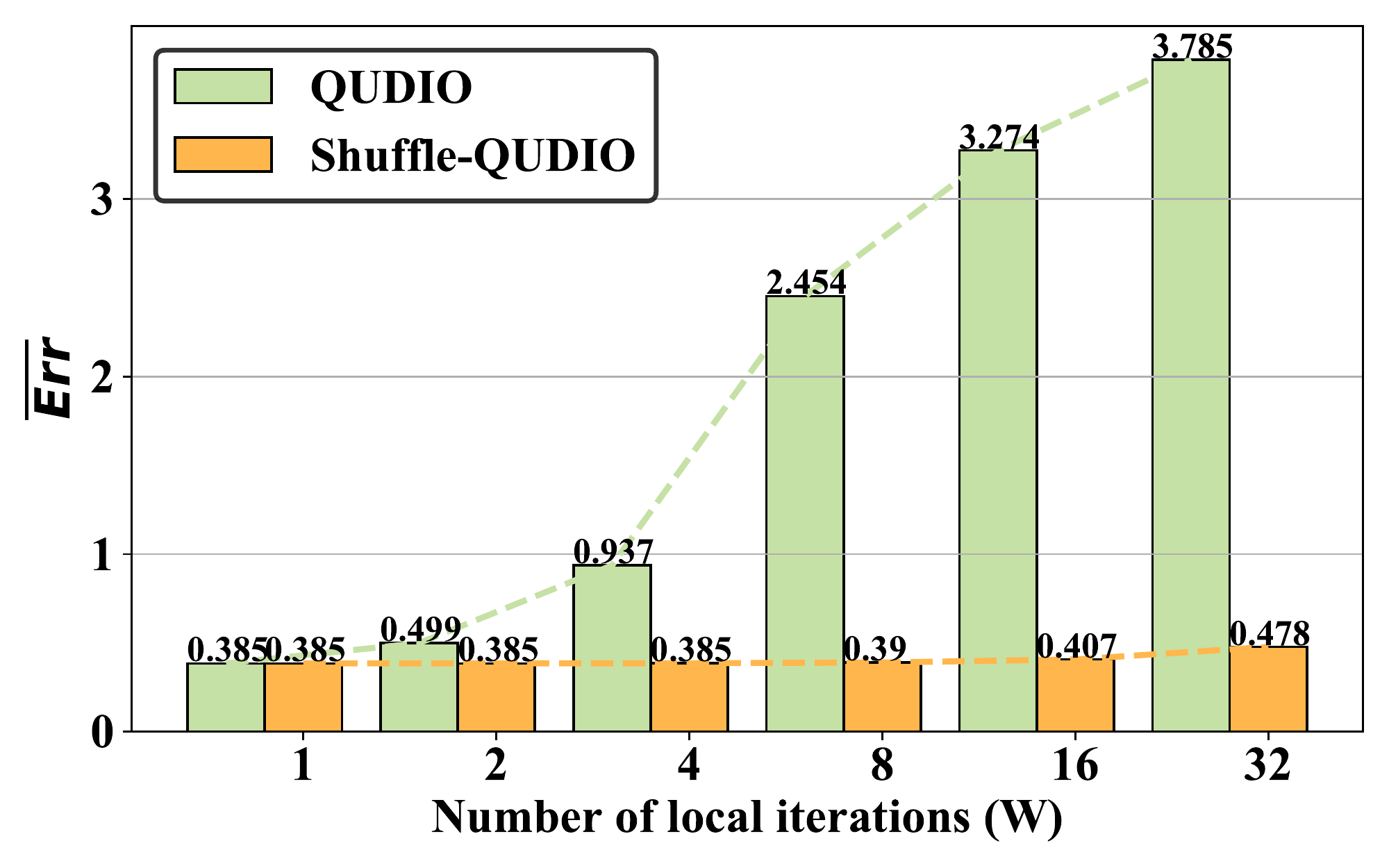}}
\vfil
\subfloat[]{
\centering
\includegraphics[width=0.45\textwidth]{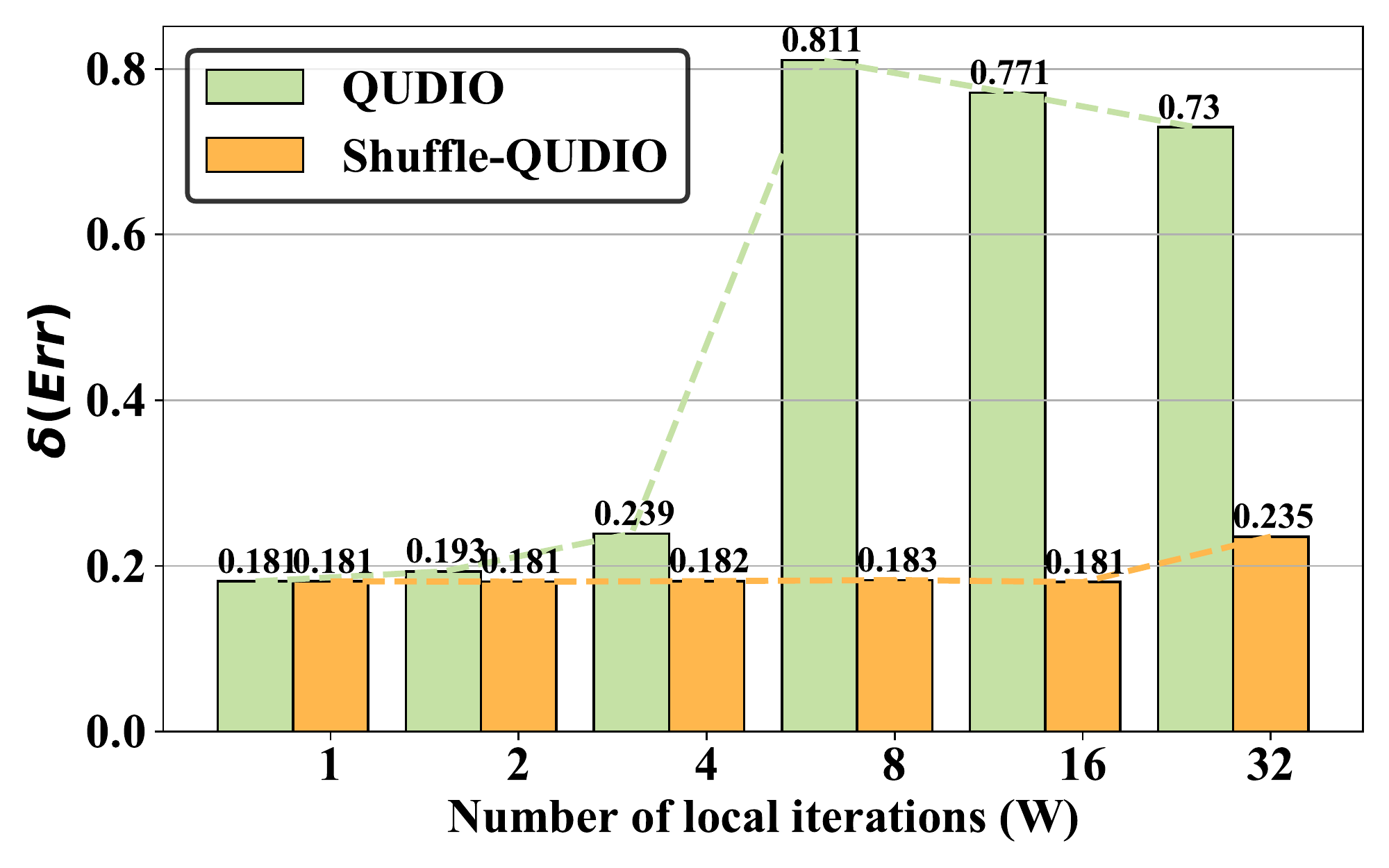}}
\caption{\small{\textbf{Mean value $\overline{Err}$ and standard deviation $\delta(Err)$ of the approximation error.} Each data point is collected over various bond distances and random seeds. Shuffle-QUDIO outperforms QUDIO in achieving smaller approximation error and lower sensitivity to communication frequency $W$.}}
\label{fig:shuffle_mean_std}
\end{figure}

To further quantify the stability of   Shuffle-QUDIO, we statistically compute the mean and standard deviation of the approximation error $Err=|E^{VQE}-E^{ideal}|$ over various bond distances and random seeds. As illustrated in the top subplot of Fig.~\ref{fig:shuffle_mean_std}, the average approximation error $\overline{Err}$ of QUDIO exponentially scales with increased $W$. When $W\geq 8$, the approximation error estimated by QUDIO exceeds $2Ha$, which fails to capture the ground state of LiH. Instead, Shuffle-QUDIO achieves an  imperceptible increment ($0.093$) of the approximation error when $W$ grows from $1$ to $32$, making it possible to largely reduce the communication overhead with a little performance drop. The bottom subplot of Fig.~\ref{fig:shuffle_mean_std} depicts the standard deviation of the approximation error derived by both methods, showing that Shuffle-QUDIO enjoys a smaller variance and a stronger stability than those of QUDIO. These observations provide the convincing empirical evidence that Shuffle-QUDIO efficiently reduces the susceptibility to $W$ in the quantum distributed optimization.

\subsection{Sensitivity to quantum noise}\label{subsec:exp-noise}
To better characterize the ability of Shuffle-QUDIO run on NISQ devices, we benchmark its performance under the depolarizing noise and the realistic quantum noise modeled by PennyLane \cite{bergholm2018pennylane}. The noise strength of the global depolarizing channel $p$ ranges from $0$ to $0.3$ with step size $0.1$. The realistic noise model is extracted from the $5$-qubit IBM \textit{ibmq\_quito} device. Note that the measurement error introduced by a finite number of shots is also considered.

\begin{figure}[htp]
\centering
\includegraphics[width=0.45\textwidth]{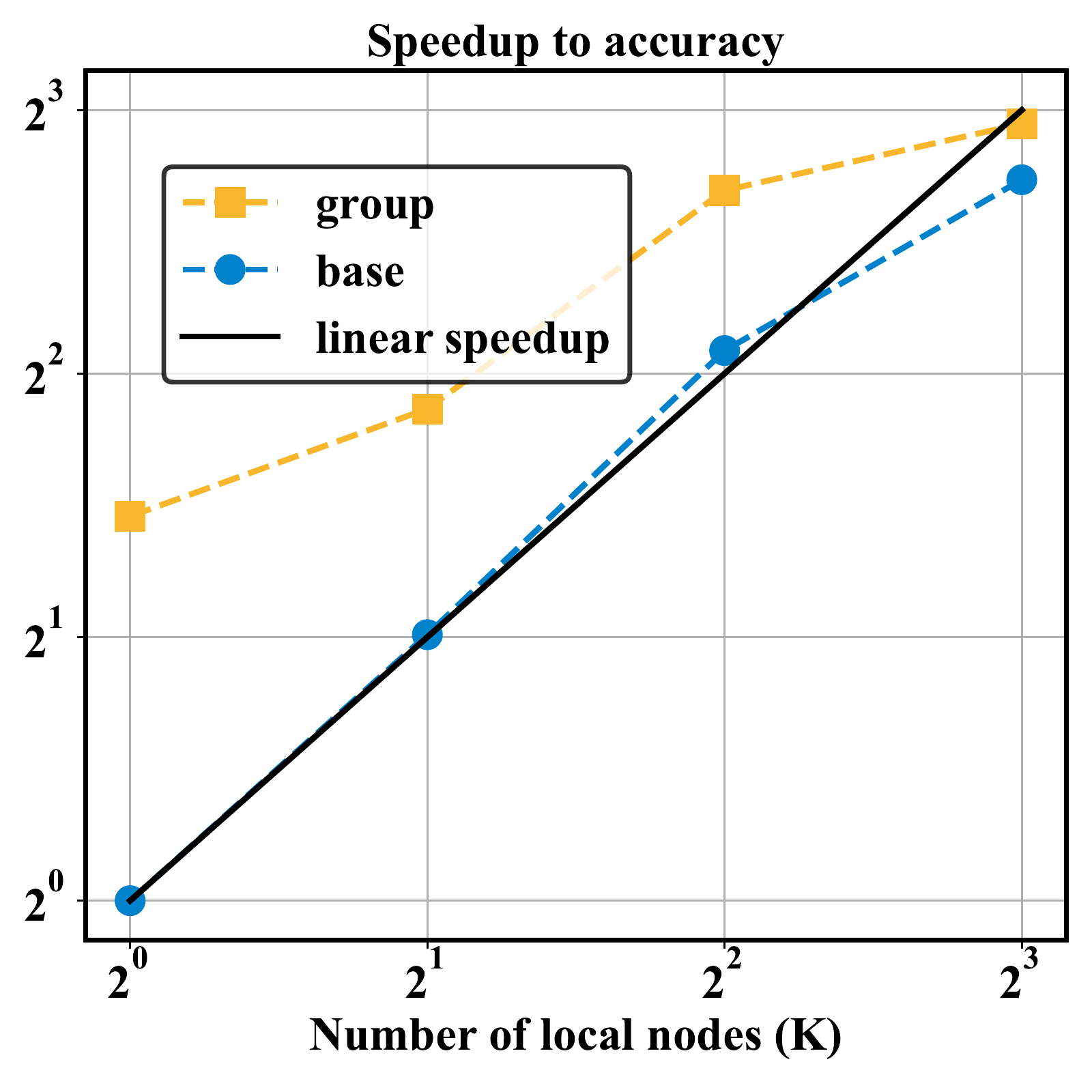}
\caption{\small{\textbf{Speedup of operator grouping to VQE for ${\rm H_2}$.} The label `\textit{base}' refers to the case that no operator grouping is applied. The label `\textit{linear speedup}' represents the reference line of linear speedup.}}
\label{fig:speedup-group}
\end{figure}

\begin{figure}[htp]
\centering
\includegraphics[width=0.45\textwidth]{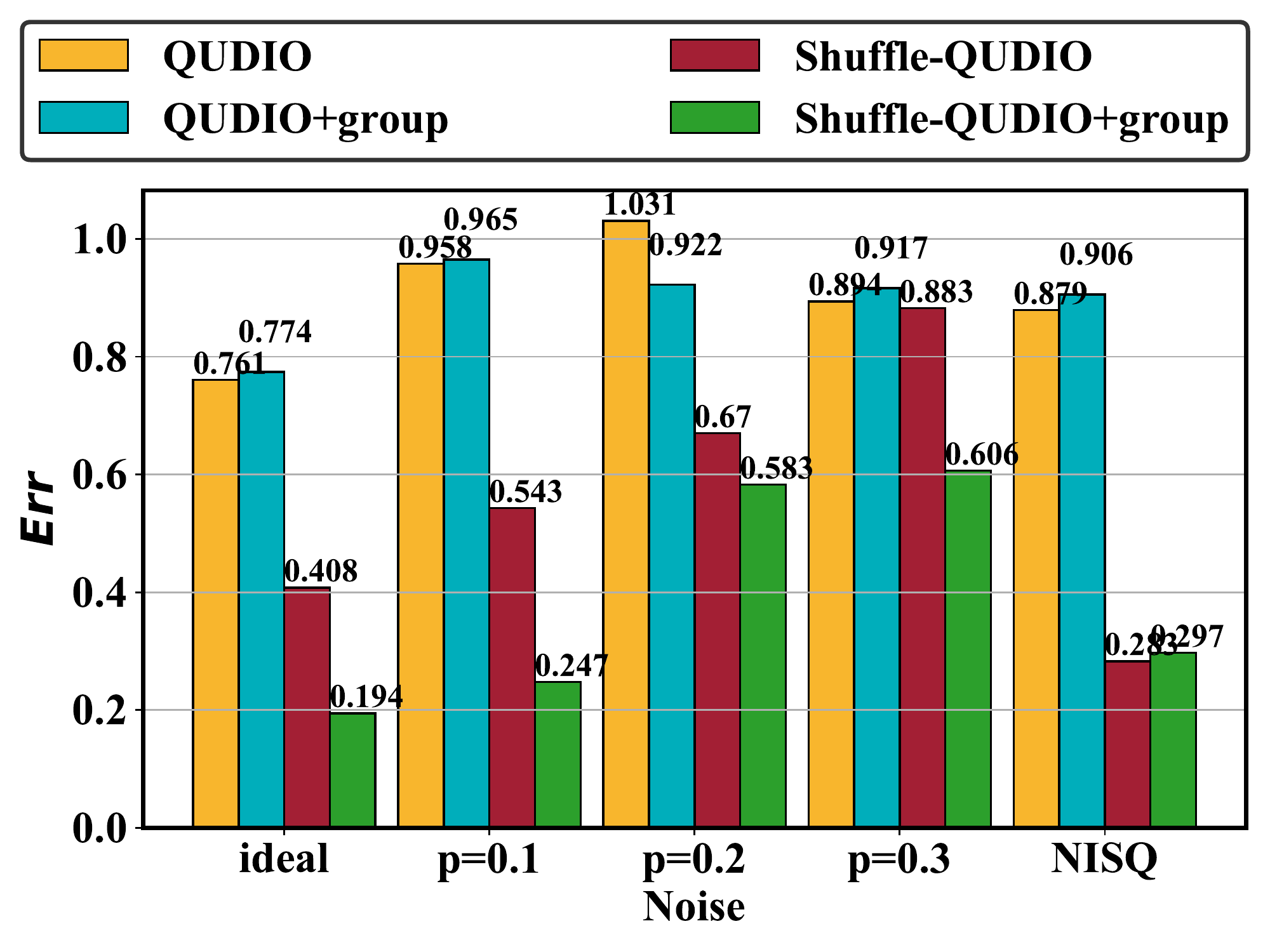}
\caption{\small{\textbf{Performance comparison in NISQ era.} \textit{ideal} represents the fault-tolerant case without noise, $p=a$ represents the case where there exists a depolarizing channel with strength $a$ in the circuit, \textit{NISQ} represents the case of running on a real NISQ device.}}
\label{fig:noise}
\end{figure}

We first benchmark the performance of Shuffle-QUDIO with the operator grouping when the shot noise is considered.  The results are shown in Fig.~\ref{fig:speedup-group}. After applying operator grouping to the molecule Hamiltonian, the trainable quantum state fast converges to the ground state of the molecule than that of the original measurement strategy. In the light of the speedup provided by the operator grouping, we can integrate this technique into the framework of Shuffle-QUDIO to gain better performance. On the other hand, with growing number of quantum processors, the acceleration rate with the operator grouping strategy gradually decays. This phenomenon partially results from the fact that a small number of Hamiltonian terms leads to a small proportion of operators that can be grouped together. 

We next apply QUDIO, QUDIO with the operator grouping, Shuffle-QUDIO, and Shuffle-QUDIO with the operator grouping to estimate the ground energy of the ${\rm H_2}$ molecule under both the system and shot noise. For each method, the hyper-parameters are set as $K=4$, $W=32$, and the number of measurements is $100$. The simulation results are shown in Fig.~\ref{fig:noise}. When the depolarizing noise is not big enough ($p<=0.2$), Shuffle-QUDIO achieves much smaller approximation error than   QUDIO. When $p=0.3$, it appears that the overwhelming noise disables both Shuffle-QUDIO and QUDIO. Under the realistic noise setting, Shuffle-QUDIO still works well with a tolerable approximation error $0.063$. By contrast, QUDIO is incapable of estimating the accurate ground state energy. Moreover, the operator grouping can further widen  the performance gap between QUDIO and Shuffle-QUDIO, by inhibiting the negative effect of quantum noise on Shuffle-QUDIO.

\subsection{Aggregation strategy}\label{sec:aggre}

Shuffle-QUDIO narrows the discrepancy among distributed processors by randomly changing the observables of each processor in every local iteration, which partially guarantee the rationality of taking average of all models for synchronization. To further explore the effect of various aggregation strategy on the performance of Shuffle-QUDIO, we devise three additional model aggregation algorithms, named as random aggregation, median aggregation and weighted aggregation.\begin{itemize}
    \item \textit{Random aggregation:} randomly select a local processor and distribute its parameters of quantum circuit to other processors.
    \item \textit{Median aggregation:} rank all local processors by their loss value and select the median as the synchronized quantum circuit.
    \item \textit{Weighted aggregation:} combine all quantum circuits of local processors by loss-induced weighted summation. The smaller the value of the loss function for a local processor, the bigger contributions the processor makes to the synchronized quantum model.
\end{itemize}
Refer to Appendix \ref{app:aggre} for more details.

\begin{figure}[htp]
\captionsetup[subfigure]{justification=centering}
\centering
\subfloat[]{
\centering
\includegraphics[width=0.45\textwidth]{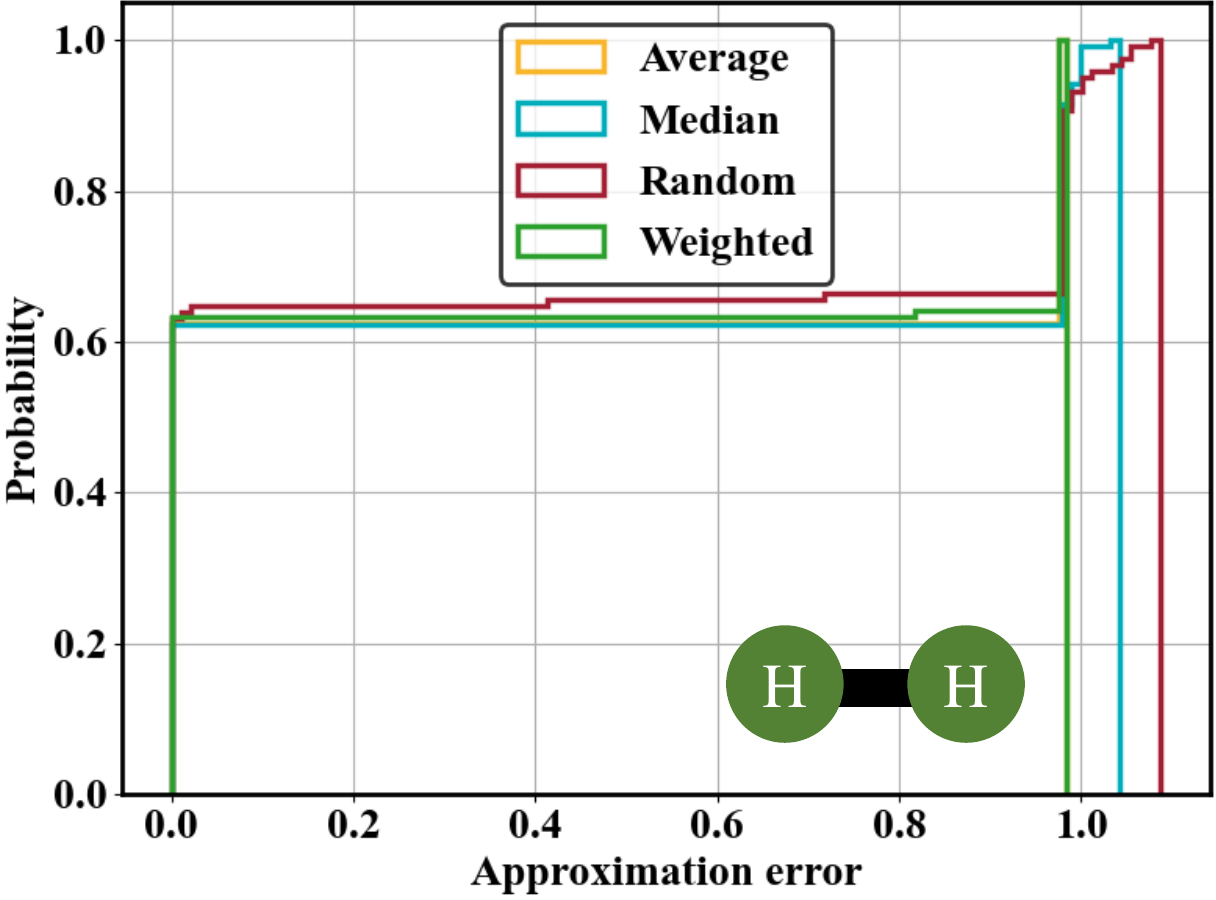}}
\vfil
\subfloat[]{
\centering
\includegraphics[width=0.45\textwidth]{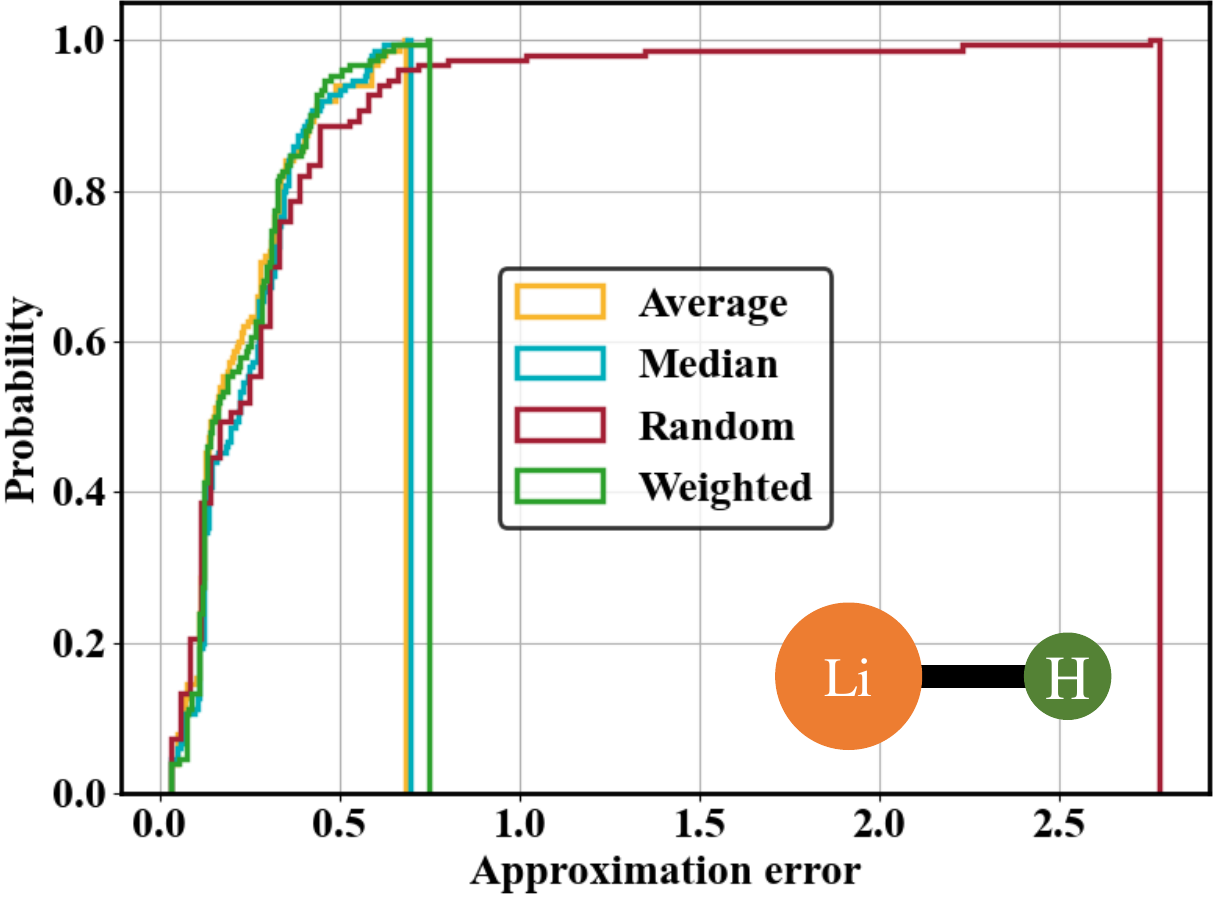}}
\caption{\small{\textbf{Performance comparison of various model aggregation algorithms. The closer the curve is to the upper left, the more accurate the estimated ground state energy is.}}}
\label{fig:aggre}
\end{figure}

We implement four quantum model aggregation methods in the framework of Shuffle-QUDIO to solve the ground state energy of molecule ${\rm H_2}$ and ${\rm LiH}$. The hyper-parameters are set as $W\in \{1,2,4,8,16,32\}$, $K\in \{1,2,4,8,16\}$. Fig.~\ref{fig:aggre} demonstrates the cumulative distribution function (CDF) of the approximation error to ground state energy. It appears no large cleavage of the approximation error among four aggregation strategies, indicating the strong robustness of Shuffle-QUDIO to quantum model aggregation. This stability may give credit to the introduction of the shuffle operation during distributed quantum computation, which diminishes the bias among different local quantum models. On the other hand, it is worth noting that average aggregation always achieves smaller approximation error with higher probability than random aggregation in the statistical sense.  This difference of CDF implies that a superior aggregation algorithm for the quantum distributed optimization could further enhance the efficiency of Shuffle-QUDIO. We leave the design of an optimal aggregation method as the future work.

\section{Discussion}\label{sec:dis}

In this paper, we propose Shuffle-QUDIO, as a novel distributed optimization scheme for VQE with faster convergence and strong noise robustness. By introducing the shuffle operation into each iteration, the Hamiltonian terms received by each local processor are not fixed during the optimization. From the statistical view, the shuffle operation warrants that the gradients manipulated by all local processors are unbiased. In this way, Shuffle-QUDIO allows an improved convergence and a lower sensitivity to communication frequency as well as quantum noise. Meanwhile, the operator group strategy can be seamlessly embedded into Shuffle-QUDIO to reduce the number of measurements in each iteration. Theoretical analysis and extensive numerical experiments on VQE further verify the effectiveness and advantages in accelerating VQE and guaranteeing small approximation errors in both ideal and noisy scenarios.

Although the random shuffle operation performs well on the $H_2$ and LiH molecules, the performance can be further improved by developing more advanced shuffling strategies. First, instead of random shuffle, we can design a problem-specific and hardware-oriented Hamiltonian allocation tactic, which can  eliminate the deviation of the optimization path of local models and better adapt to the limited quantum resources of various local processors. Second, due to the existence of barren plateau \cite{mcclean2018barren,marrero2021entanglement,wang2021noise,cerezo2021cost,arrasmith2021effect} in the optimization of ansatz, the training of local quantum models may get stuck. Inspired by the study that local observables enjoy a polynomially vanishing gradient \cite{cerezo2021cost}, a promising direction is to group Hamiltonian terms with similar locality in QUDIO to avoid the barren plateau of some processors. Finally, a more fine-grained partition of the quantum circuit structure besides observables can be employed to reduce the number of parameters to be optimized for each local processor, as implemented in \cite{zhang2022variational}.

Another feasible attempt to enhance the performance of distributed VQE in practice is to unify Shuffle-QUDIO with other measurement reduction techniques. One successful example is operator grouping, as discussed in Section \ref{subsec:exp-noise}. Specifically, when optimizing the circuit run on each distributed quantum processor, we can utilize the operator grouping strategy to reduce the required number of measurements of the allocated Hamiltonian. In this way, the measurement noise in the framework of Shuffle-QUDIO is eliminated under finite budget of shot number. Other two methods, like shot allocation and classical shadows, can be also integrated into Shuffle-QUDIO in the similar manner.

Besides the potential improvements in convergence and speedup for Shuffle-QUDIO, the data privacy leakage during transmitting gradient information among local nodes should be avoided. One the one hand, the shuffle operation in Shuffle-QUDIO naturally adds randomness to the system, hindering the recovery of intact data. On the other hand, previous studies proposed differential privacy \cite{du2020quantum,du2021quantum} and blind quantum computing \cite{li2021quantum} to protect data security. When combining these techniques and Shuffle-QUDIO, it remains open to explore the consequent influence on the convergence of optimization.

Apart from utilizing the quantum-specific properties to enhance Shuffle-QUDIO, we can also leverage the experience from classical distributed optimization, such as Elastic Averaging SGD \cite{zhang2015deep}, decentralized SGD \cite{koloskova2020unified}. It is worth noting that the flexibility of Shuffle-QUDIO makes it easy to replace some components with advanced classical techniques, as discussed in Sec.~\ref{sec:aggre}. Taken together, it is expected to utilize Shuffle-QUDIO and its variants to speed up the computation of variational quantum circuits and tackle real-world problems with NISQ devices.

\medskip 
\textbf{Code availability.} The source code of Shuffle-QUDIO open-source is available at \url{https://github.com/QQQYang/Shuffle-QUDIO}.


\begin{thebibliography}{10}

\bibitem{spring2013boson}
Justin~B Spring, Benjamin~J Metcalf, Peter~C Humphreys, W~Steven Kolthammer,
  Xian-Min Jin, Marco Barbieri, Animesh Datta, Nicholas Thomas-Peter, Nathan~K
  Langford, Dmytro Kundys, et~al.
\newblock Boson sampling on a photonic chip.
\newblock {\em Science}, 339(6121):798--801, 2013.

\bibitem{wang2017high}
Hui Wang, Yu~He, Yu-Huai Li, Zu-En Su, Bo~Li, He-Liang Huang, Xing Ding,
  Ming-Cheng Chen, Chang Liu, Jian Qin, et~al.
\newblock High-efficiency multiphoton boson sampling.
\newblock {\em Nature Photonics}, 11(6):361--365, 2017.

\bibitem{bulmer2021boundary}
Jacob~FF Bulmer, Bryn~A Bell, Rachel~S Chadwick, Alex~E Jones, Diana Moise,
  Alessandro Rigazzi, Jan Thorbecke, Utz-Uwe Haus, Thomas Van~Vaerenbergh,
  Raj~B Patel, et~al.
\newblock The boundary for quantum advantage in gaussian boson sampling.
\newblock {\em Science Advances}, 8(4):eabl9236, 2021.

\bibitem{jiang2018quantum}
Shuxian Jiang, Keith~A Britt, Alexander~J McCaskey, Travis~S Humble, and Sabre
  Kais.
\newblock Quantum annealing for prime factorization.
\newblock {\em Scientific reports}, 8(1):1--9, 2018.

\bibitem{peng2019factoring}
WangChun Peng, BaoNan Wang, Feng Hu, YunJiang Wang, XianJin Fang, XingYuan
  Chen, and Chao Wang.
\newblock Factoring larger integers with fewer qubits via quantum annealing
  with optimized parameters.
\newblock {\em SCIENCE CHINA Physics, Mechanics \& Astronomy}, 62(6):1--8,
  2019.

\bibitem{preskill2018quantum}
John Preskill.
\newblock Quantum computing in the nisq era and beyond.
\newblock {\em Quantum}, 2:79, 2018.

\bibitem{bharti2021noisy}
Kishor Bharti, Alba Cervera-Lierta, Thi~Ha Kyaw, Tobias Haug, Sumner
  Alperin-Lea, Abhinav Anand, Matthias Degroote, Hermanni Heimonen, Jakob~S
  Kottmann, Tim Menke, et~al.
\newblock Noisy intermediate-scale quantum (nisq) algorithms.
\newblock {\em arXiv preprint arXiv:2101.08448}, 2021.

\bibitem{arute2019quantum}
Frank Arute, Kunal Arya, Ryan Babbush, Dave Bacon, Joseph~C Bardin, Rami
  Barends, Rupak Biswas, Sergio Boixo, Fernando~GSL Brandao, David~A Buell,
  et~al.
\newblock Quantum supremacy using a programmable superconducting processor.
\newblock {\em Nature}, 574(7779):505--510, 2019.

\bibitem{zhu2021quantum}
Qingling Zhu, Sirui Cao, Fusheng Chen, Ming-Cheng Chen, Xiawei Chen, Tung-Hsun
  Chung, Hui Deng, Yajie Du, Daojin Fan, Ming Gong, et~al.
\newblock Quantum computational advantage via 60-qubit 24-cycle random circuit
  sampling.
\newblock {\em arXiv preprint arXiv:2109.03494}, 2021.

\bibitem{mcclean2016theory}
Jarrod~R McClean, Jonathan Romero, Ryan Babbush, and Al{\'a}n Aspuru-Guzik.
\newblock The theory of variational hybrid quantum-classical algorithms.
\newblock {\em New Journal of Physics}, 18(2):023023, 2016.

\bibitem{cerezo2021variational}
Marco Cerezo, Andrew Arrasmith, Ryan Babbush, Simon~C Benjamin, Suguru Endo,
  Keisuke Fujii, Jarrod~R McClean, Kosuke Mitarai, Xiao Yuan, Lukasz Cincio,
  et~al.
\newblock Variational quantum algorithms.
\newblock {\em Nature Reviews Physics}, 3(9):625--644, 2021.

\bibitem{cerezo2020variational}
Marco Cerezo, Alexander Poremba, Lukasz Cincio, and Patrick~J Coles.
\newblock Variational quantum fidelity estimation.
\newblock {\em Quantum}, 4:248, 2020.

\bibitem{qian2021dilemma}
Yang Qian, Xinbiao Wang, Yuxuan Du, Xingyao Wu, and Dacheng Tao.
\newblock The dilemma of quantum neural networks.
\newblock {\em arXiv preprint arXiv:2106.04975}, 2021.

\bibitem{tian2022recent}
Jinkai Tian, Xiaoyu Sun, Yuxuan Du, Shanshan Zhao, Qing Liu, Kaining Zhang, Wei
  Yi, Wanrong Huang, Chaoyue Wang, Xingyao Wu, et~al.
\newblock Recent advances for quantum neural networks in generative learning.
\newblock {\em arXiv preprint arXiv:2206.03066}, 2022.

\bibitem{orus2019quantum}
Roman Orus, Samuel Mugel, and Enrique Lizaso.
\newblock Quantum computing for finance: Overview and prospects.
\newblock {\em Reviews in Physics}, 4:100028, 2019.

\bibitem{pistoia2021quantum}
Marco Pistoia, Syed~Farhan Ahmad, Akshay Ajagekar, Alexander Buts, Shouvanik
  Chakrabarti, Dylan Herman, Shaohan Hu, Andrew Jena, Pierre Minssen, Pradeep
  Niroula, et~al.
\newblock Quantum machine learning for finance iccad special session paper.
\newblock In {\em 2021 IEEE/ACM International Conference On Computer Aided
  Design (ICCAD)}, pages 1--9. IEEE, 2021.

\bibitem{grimsley2019adaptive}
Harper~R Grimsley, Sophia~E Economou, Edwin Barnes, and Nicholas~J Mayhall.
\newblock An adaptive variational algorithm for exact molecular simulations on
  a quantum computer.
\newblock {\em Nature communications}, 10(1):1--9, 2019.

\bibitem{arute2020hartree}
Frank Arute, Kunal Arya, Ryan Babbush, Dave Bacon, Joseph~C Bardin, Rami
  Barends, Sergio Boixo, Michael Broughton, Bob~B Buckley, David~A Buell,
  et~al.
\newblock Hartree-fock on a superconducting qubit quantum computer.
\newblock {\em Science}, 369(6507):1084--1089, 2020.

\bibitem{kandala2017hardware}
Abhinav Kandala, Antonio Mezzacapo, Kristan Temme, Maika Takita, Markus Brink,
  Jerry~M Chow, and Jay~M Gambetta.
\newblock Hardware-efficient variational quantum eigensolver for small
  molecules and quantum magnets.
\newblock {\em Nature}, 549(7671):242--246, 2017.

\bibitem{robert2021resource}
Anton Robert, Panagiotis~Kl Barkoutsos, Stefan Woerner, and Ivano Tavernelli.
\newblock Resource-efficient quantum algorithm for protein folding.
\newblock {\em npj Quantum Information}, 7(1):1--5, 2021.

\bibitem{kais2014introduction}
Sabre Kais.
\newblock Introduction to quantum information and computation for chemistry.
\newblock {\em Quantum Information and Computation for Chemistry}, pages 1--38,
  2014.

\bibitem{wecker2015solving}
Dave Wecker, Matthew~B Hastings, Nathan Wiebe, Bryan~K Clark, Chetan Nayak, and
  Matthias Troyer.
\newblock Solving strongly correlated electron models on a quantum computer.
\newblock {\em Physical Review A}, 92(6):062318, 2015.

\bibitem{cai2020quantum}
Xiaoxia Cai, Wei-Hai Fang, Heng Fan, and Zhendong Li.
\newblock Quantum computation of molecular response properties.
\newblock {\em Physical Review Research}, 2(3):033324, 2020.

\bibitem{wang2019accelerated}
Daochen Wang, Oscar Higgott, and Stephen Brierley.
\newblock Accelerated variational quantum eigensolver.
\newblock {\em Physical review letters}, 122(14):140504, 2019.

\bibitem{romero2018strategies}
Jonathan Romero, Ryan Babbush, Jarrod~R McClean, Cornelius Hempel, Peter~J
  Love, and Al{\'a}n Aspuru-Guzik.
\newblock Strategies for quantum computing molecular energies using the unitary
  coupled cluster ansatz.
\newblock {\em Quantum Science and Technology}, 4(1):014008, 2018.

\bibitem{Cervera2021meta-variational}
Alba Cervera-Lierta, Jakob~S. Kottmann, and Al\'an Aspuru-Guzik.
\newblock Meta-variational quantum eigensolver: Learning energy profiles of
  parameterized hamiltonians for quantum simulation.
\newblock {\em PRX Quantum}, 2:020329, May 2021.

\bibitem{parrish2019quantum}
Robert~M Parrish, Edward~G Hohenstein, Peter~L McMahon, and Todd~J
  Mart{\'\i}nez.
\newblock Quantum computation of electronic transitions using a variational
  quantum eigensolver.
\newblock {\em Physical review letters}, 122(23):230401, 2019.

\bibitem{huang2021provably}
Hsin-Yuan Huang, Richard Kueng, Giacomo Torlai, Victor~V Albert, and John
  Preskill.
\newblock Provably efficient machine learning for quantum many-body problems.
\newblock {\em arXiv preprint arXiv:2106.12627}, 2021.

\bibitem{lee2021neural}
Chee~Kong Lee, Pranay Patil, Shengyu Zhang, and Chang~Yu Hsieh.
\newblock Neural-network variational quantum algorithm for simulating many-body
  dynamics.
\newblock {\em Physical Review Research}, 3(2):023095, 2021.

\bibitem{endo2020variational}
Suguru Endo, Jinzhao Sun, Ying Li, Simon~C Benjamin, and Xiao Yuan.
\newblock Variational quantum simulation of general processes.
\newblock {\em Physical Review Letters}, 125(1):010501, 2020.

\bibitem{huang2021power}
Hsin-Yuan Huang, Michael Broughton, Masoud Mohseni, Ryan Babbush, Sergio Boixo,
  Hartmut Neven, and Jarrod~R McClean.
\newblock Power of data in quantum machine learning.
\newblock {\em Nature communications}, 12(1):1--9, 2021.

\bibitem{huang2022quantum}
Hsin-Yuan Huang, Michael Broughton, Jordan Cotler, Sitan Chen, Jerry Li, Masoud
  Mohseni, Hartmut Neven, Ryan Babbush, Richard Kueng, John Preskill, et~al.
\newblock Quantum advantage in learning from experiments.
\newblock {\em Science}, 376(6598):1182--1186, 2022.

\bibitem{du2021exploring}
Yuxuan Du and Dacheng Tao.
\newblock On exploring practical potentials of quantum auto-encoder with
  advantages.
\newblock {\em arXiv preprint arXiv:2106.15432}, 2021.

\bibitem{caro2022generalization}
Matthias~C Caro, Hsin-Yuan Huang, Kunal Sharma, Andrew Sornborger, Lukasz
  Cincio, and Patrick~J Coles.
\newblock Generalization in quantum machine learning from few training data.
\newblock {\em Nature communications}, 13(1):1--11, 2022.

\bibitem{gili2022quantum}
Kaitlin Gili, Mohamed Hibat-Allah, Marta Mauri, Chris Ballance, and Alejandro
  Perdomo-Ortiz.
\newblock Do quantum circuit born machines generalize?
\newblock {\em arXiv preprint arXiv:2207.13645}, 2022.

\bibitem{farhi2014quantum}
Edward Farhi, Jeffrey Goldstone, and Sam Gutmann.
\newblock A quantum approximate optimization algorithm.
\newblock {\em arXiv preprint arXiv:1411.4028}, 2014.

\bibitem{zhou2020quantum}
Leo Zhou, Sheng-Tao Wang, Soonwon Choi, Hannes Pichler, and Mikhail~D Lukin.
\newblock Quantum approximate optimization algorithm: Performance, mechanism,
  and implementation on near-term devices.
\newblock {\em Physical Review X}, 10(2):021067, 2020.

\bibitem{harrigan2021quantum}
Matthew~P Harrigan, Kevin~J Sung, Matthew Neeley, Kevin~J Satzinger, Frank
  Arute, Kunal Arya, Juan Atalaya, Joseph~C Bardin, Rami Barends, Sergio Boixo,
  et~al.
\newblock Quantum approximate optimization of non-planar graph problems on a
  planar superconducting processor.
\newblock {\em Nature Physics}, 17(3):332--336, 2021.

\bibitem{lacroix2020improving}
Nathan Lacroix, Christoph Hellings, Christian~Kraglund Andersen, Agustin
  Di~Paolo, Ants Remm, Stefania Lazar, Sebastian Krinner, Graham~J Norris,
  Mihai Gabureac, Johannes Heinsoo, et~al.
\newblock Improving the performance of deep quantum optimization algorithms
  with continuous gate sets.
\newblock {\em PRX Quantum}, 1(2):110304, 2020.

\bibitem{hadfield2019quantum}
Stuart Hadfield, Zhihui Wang, Bryan O’Gorman, Eleanor~G Rieffel, Davide
  Venturelli, and Rupak Biswas.
\newblock From the quantum approximate optimization algorithm to a quantum
  alternating operator ansatz.
\newblock {\em Algorithms}, 12(2):34, 2019.

\bibitem{zhou2022qaoa}
Zeqiao Zhou, Yuxuan Du, Xinmei Tian, and Dacheng Tao.
\newblock Qaoa-in-qaoa: solving large-scale maxcut problems on small quantum
  machines.
\newblock {\em arXiv preprint arXiv:2205.11762}, 2022.

\bibitem{peruzzo2014variational}
Alberto Peruzzo, Jarrod McClean, Peter Shadbolt, Man-Hong Yung, Xiao-Qi Zhou,
  Peter~J Love, Al{\'a}n Aspuru-Guzik, and Jeremy~L O’brien.
\newblock A variational eigenvalue solver on a photonic quantum processor.
\newblock {\em Nature communications}, 5:4213, 2014.

\bibitem{seeley2012bravyi}
Jacob~T Seeley, Martin~J Richard, and Peter~J Love.
\newblock The bravyi-kitaev transformation for quantum computation of
  electronic structure.
\newblock {\em The Journal of chemical physics}, 137(22):224109, 2012.

\bibitem{bravyi2002fermionic}
Sergey~B Bravyi and Alexei~Yu Kitaev.
\newblock Fermionic quantum computation.
\newblock {\em Annals of Physics}, 298(1):210--226, 2002.

\bibitem{jordan1993paulische}
Pascual Jordan and Eugene~Paul Wigner.
\newblock {\"u}ber das paulische {\"a}quivalenzverbot.
\newblock In {\em The Collected Works of Eugene Paul Wigner}, pages 109--129.
  Springer, 1993.

\bibitem{gonthier2020identifying}
J{\'e}r{\^o}me~F Gonthier, Maxwell~D Radin, Corneliu Buda, Eric~J Doskocil,
  Clena~M Abuan, and Jhonathan Romero.
\newblock Identifying challenges towards practical quantum advantage through
  resource estimation: the measurement roadblock in the variational quantum
  eigensolver.
\newblock {\em arXiv preprint arXiv:2012.04001}, 2020.

\bibitem{ralliImplementationMeasurementReduction2020}
Alexis Ralli, Peter Love, Andrew Tranter, and Peter Coveney.
\newblock Implementation of {{Measurement Reduction}} for the {{Variational
  Quantum Eigensolver}}.
\newblock {\em arXiv:2012.02765 [physics, physics:quant-ph]}, December 2020.

\bibitem{verteletskyiMeasurementOptimizationVariational2020}
Vladyslav Verteletskyi, Tzu-Ching Yen, and Artur~F. Izmaylov.
\newblock Measurement {{Optimization}} in the {{Variational Quantum Eigensolver
  Using}} a {{Minimum Clique Cover}}.
\newblock {\em The Journal of Chemical Physics}, 152(12):124114, March 2020.

\bibitem{zhao2020measurement}
Andrew Zhao, Andrew Tranter, William~M Kirby, Shu~Fay Ung, Akimasa Miyake, and
  Peter~J Love.
\newblock Measurement reduction in variational quantum algorithms.
\newblock {\em Physical Review A}, 101(6):062322, 2020.

\bibitem{gokhaleMinimizingStatePreparations2019}
Pranav Gokhale, Olivia Angiuli, Yongshan Ding, Kaiwen Gui, Teague Tomesh,
  Martin Suchara, Margaret Martonosi, and Frederic~T. Chong.
\newblock Minimizing {{State Preparations}} in {{Variational Quantum
  Eigensolver}} by {{Partitioning}} into {{Commuting Families}}.
\newblock {\em arXiv:1907.13623 [quant-ph]}, July 2019.

\bibitem{tkachenko2021correlation}
Nikolay~V Tkachenko, James Sud, Yu~Zhang, Sergei Tretiak, Petr~M Anisimov,
  Andrew~T Arrasmith, Patrick~J Coles, Lukasz Cincio, and Pavel~A Dub.
\newblock Correlation-informed permutation of qubits for reducing ansatz depth
  in the variational quantum eigensolver.
\newblock {\em PRX Quantum}, 2(2):020337, 2021.

\bibitem{zhang2022variational}
Yu~Zhang, Lukasz Cincio, Christian~FA Negre, Piotr Czarnik, Patrick~J Coles,
  Petr~M Anisimov, Susan~M Mniszewski, Sergei Tretiak, and Pavel~A Dub.
\newblock Variational quantum eigensolver with reduced circuit complexity.
\newblock {\em npj Quantum Information}, 8(1):1--10, 2022.

\bibitem{arrasmith2020operator}
Andrew Arrasmith, Lukasz Cincio, Rolando~D Somma, and Patrick~J Coles.
\newblock Operator sampling for shot-frugal optimization in variational
  algorithms.
\newblock {\em arXiv preprint arXiv:2004.06252}, 2020.

\bibitem{van2021measurement}
Barnaby van Straaten and B{\'a}lint Koczor.
\newblock Measurement cost of metric-aware variational quantum algorithms.
\newblock {\em PRX Quantum}, 2(3):030324, 2021.

\bibitem{gu2021adaptive}
Andi Gu, Angus Lowe, Pavel~A Dub, Patrick~J Coles, and Andrew Arrasmith.
\newblock Adaptive shot allocation for fast convergence in variational quantum
  algorithms.
\newblock {\em arXiv preprint arXiv:2108.10434}, 2021.

\bibitem{huang2020predicting}
Hsin-Yuan Huang, Richard Kueng, and John Preskill.
\newblock Predicting many properties of a quantum system from very few
  measurements.
\newblock {\em Nature Physics}, 16(10):1050--1057, 2020.

\bibitem{hadfield2022measurements}
Charles Hadfield, Sergey Bravyi, Rudy Raymond, and Antonio Mezzacapo.
\newblock Measurements of quantum hamiltonians with locally-biased classical
  shadows.
\newblock {\em Communications in Mathematical Physics}, 391(3):951--967, 2022.

\bibitem{andres2019automated}
Pablo Andres-Martinez and Chris Heunen.
\newblock Automated distribution of quantum circuits via hypergraph
  partitioning.
\newblock {\em Physical Review A}, 100(3):032308, 2019.

\bibitem{barratt2003parallel}
F~Barratt, J~Dborin, M~Bal, V~Stojevic, F~Pollmann, and AG~Green.
\newblock Parallel quantum simulation of large systems on small quantum
  computers (2020).
\newblock {\em arXiv preprint arXiv:2003.12087}, 2003.

\bibitem{duAcceleratingVariationalQuantum2021}
Yuxuan Du, Yang Qian, Xingyao Wu, and Dacheng Tao.
\newblock A distributed learning scheme for variational quantum algorithms.
\newblock {\em IEEE Transactions on Quantum Engineering}, 2022.

\bibitem{mineh2022accelerating}
Lana Mineh and Ashley Montanaro.
\newblock Accelerating the variational quantum eigensolver using parallelism.
\newblock {\em arXiv preprint arXiv:2209.03796}, 2022.

\bibitem{tang2021qubit}
Ho~Lun Tang, VO~Shkolnikov, George~S Barron, Harper~R Grimsley, Nicholas~J
  Mayhall, Edwin Barnes, and Sophia~E Economou.
\newblock qubit-adapt-vqe: An adaptive algorithm for constructing
  hardware-efficient ans{\"a}tze on a quantum processor.
\newblock {\em PRX Quantum}, 2(2):020310, 2021.

\bibitem{barratt2020parallel}
Fergus Barratt, James Dborin, Matthias Bal, Vid Stojevic, Frank Pollmann, and
  Andrew~G Green.
\newblock Parallel quantum simulation of large systems on small quantum
  computers.
\newblock {\em arXiv preprint arXiv:2003.12087}, 2020.

\bibitem{diadamo2021distributed}
Stephen DiAdamo, Marco Ghibaudi, and James Cruise.
\newblock Distributed quantum computing and network control for accelerated
  vqe.
\newblock {\em arXiv preprint arXiv:2101.02504}, 2021.

\bibitem{lecun1988theoretical}
Yann LeCun, D~Touresky, G~Hinton, and T~Sejnowski.
\newblock A theoretical framework for back-propagation.
\newblock In {\em Proceedings of the 1988 connectionist models summer school},
  volume~1, pages 21--28, 1988.

\bibitem{banchi2021measuring}
Leonardo Banchi and Gavin~E Crooks.
\newblock Measuring analytic gradients of general quantum evolution with the
  stochastic parameter shift rule.
\newblock {\em Quantum}, 5:386, 2021.

\bibitem{wierichs2022general}
David Wierichs, Josh Izaac, Cody Wang, and Cedric Yen-Yu Lin.
\newblock General parameter-shift rules for quantum gradients.
\newblock {\em Quantum}, 6:677, 2022.

\bibitem{du2020learnability}
Yuxuan Du, Min-Hsiu Hsieh, Tongliang Liu, Shan You, and Dacheng Tao.
\newblock Learnability of quantum neural networks.
\newblock {\em PRX Quantum}, 2(4):040337, 2021.

\bibitem{endo2018practical}
Suguru Endo, Simon~C Benjamin, and Ying Li.
\newblock Practical quantum error mitigation for near-future applications.
\newblock {\em Physical Review X}, 8(3):031027, 2018.

\bibitem{endo2021hybrid}
Suguru Endo, Zhenyu Cai, Simon~C Benjamin, and Xiao Yuan.
\newblock Hybrid quantum-classical algorithms and quantum error mitigation.
\newblock {\em Journal of the Physical Society of Japan}, 90(3):032001, 2021.

\bibitem{strikis2021learning}
Armands Strikis, Dayue Qin, Yanzhu Chen, Simon~C Benjamin, and Ying Li.
\newblock Learning-based quantum error mitigation.
\newblock {\em PRX Quantum}, 2(4):040330, 2021.

\bibitem{du2020quantumsearch}
Yuxuan Du, Tao Huang, Shan You, Min-Hsiu Hsieh, and Dacheng Tao.
\newblock Quantum circuit architecture search for variational quantum
  algorithms.
\newblock {\em npj Quantum Information}, 8(1):1--8, 2022.

\bibitem{haddadpour2019local}
Farzin Haddadpour, Mohammad~Mahdi Kamani, Mehrdad Mahdavi, and Viveck Cadambe.
\newblock Local sgd with periodic averaging: Tighter analysis and adaptive
  synchronization.
\newblock {\em Advances in Neural Information Processing Systems}, 32, 2019.

\bibitem{sweke2020stochastic}
Ryan Sweke, Frederik Wilde, Johannes Meyer, Maria Schuld, Paul~K F{\"a}hrmann,
  Barth{\'e}l{\'e}my Meynard-Piganeau, and Jens Eisert.
\newblock Stochastic gradient descent for hybrid quantum-classical
  optimization.
\newblock {\em Quantum}, 4:314, 2020.

\bibitem{bergholm2018pennylane}
Ville Bergholm, Josh Izaac, Maria Schuld, Christian Gogolin, M~Sohaib Alam,
  Shahnawaz Ahmed, Juan~Miguel Arrazola, Carsten Blank, Alain Delgado, Soran
  Jahangiri, et~al.
\newblock Pennylane: Automatic differentiation of hybrid quantum-classical
  computations.
\newblock {\em arXiv preprint arXiv:1811.04968}, 2018.

\bibitem{mcclean2018barren}
Jarrod~R McClean, Sergio Boixo, Vadim~N Smelyanskiy, Ryan Babbush, and Hartmut
  Neven.
\newblock Barren plateaus in quantum neural network training landscapes.
\newblock {\em Nature communications}, 9(1):1--6, 2018.

\bibitem{marrero2021entanglement}
Carlos~Ortiz Marrero, M{\'a}ria Kieferov{\'a}, and Nathan Wiebe.
\newblock Entanglement-induced barren plateaus.
\newblock {\em PRX Quantum}, 2(4):040316, 2021.

\bibitem{wang2021noise}
Samson Wang, Enrico Fontana, Marco Cerezo, Kunal Sharma, Akira Sone, Lukasz
  Cincio, and Patrick~J Coles.
\newblock Noise-induced barren plateaus in variational quantum algorithms.
\newblock {\em Nature communications}, 12(1):1--11, 2021.

\bibitem{cerezo2021cost}
Marco Cerezo, Akira Sone, Tyler Volkoff, Lukasz Cincio, and Patrick~J Coles.
\newblock Cost function dependent barren plateaus in shallow parametrized
  quantum circuits.
\newblock {\em Nature communications}, 12(1):1--12, 2021.

\bibitem{arrasmith2021effect}
Andrew Arrasmith, M~Cerezo, Piotr Czarnik, Lukasz Cincio, and Patrick~J Coles.
\newblock Effect of barren plateaus on gradient-free optimization.
\newblock {\em Quantum}, 5:558, 2021.

\bibitem{du2020quantum}
Yuxuan Du, Min-Hsiu Hsieh, Tongliang Liu, Shan You, and Dacheng Tao.
\newblock Quantum differentially private sparse regression learning.
\newblock {\em IEEE Transactions on Information Theory}, 2022.

\bibitem{du2021quantum}
Yuxuan Du, Min-Hsiu Hsieh, Tongliang Liu, Dacheng Tao, and Nana Liu.
\newblock Quantum noise protects quantum classifiers against adversaries.
\newblock {\em Physical Review Research}, 3(2):023153, 2021.

\bibitem{li2021quantum}
Weikang Li, Sirui Lu, and Dong-Ling Deng.
\newblock Quantum private distributed learning through blind quantum computing.
\newblock {\em arXiv preprint arXiv:2103.08403}, 2021.

\bibitem{zhang2015deep}
Sixin Zhang, Anna~E Choromanska, and Yann LeCun.
\newblock Deep learning with elastic averaging sgd.
\newblock {\em Advances in neural information processing systems}, 28, 2015.

\bibitem{koloskova2020unified}
Anastasia Koloskova, Nicolas Loizou, Sadra Boreiri, Martin Jaggi, and Sebastian
  Stich.
\newblock A unified theory of decentralized sgd with changing topology and
  local updates.
\newblock In {\em International Conference on Machine Learning}, pages
  5381--5393. PMLR, 2020.

\end{thebibliography}

\newpage   
\clearpage 
\medskip

\newpage   
\clearpage 
\appendix 
\onecolumngrid

The appendix is organized as follows. Appendix \ref{app:notation} introduces basic notations and properties of the loss function. Appendices \ref{pr:shuffle}, \ref{pr:grad-bias}, and \ref{app:proof1} present the proofs  of Lemma \ref{the:shuffle}, Lemma \ref{lemma:grad-bias}, and Theorem \ref{the:convergence}, respectively. Appendix \ref{app:aggre} explains the aggregation methods discussed in Section \ref{sec:aggre}. Appendix \ref{app:traj} demonstrates the  additional experiment results and analysis about Shuffle-QUDDIO.

\section{Notations and properties of the loss function}\label{app:notation}

\subsection{Notations}

The notations are unified as follows. We denote $W$ as the number of local iterations between the global synchronization, $K$ as the number of quantum processors, and $T$ as the number of total iterations. During the optimization, we denote $\bm{\theta}$ as the collection of trainable parameters, $\bm{\theta}_k^t$ as the parameters of the $k$-th processor at the $t$ iteration, $\eta$ as the learning rate, $L$ as the loss function, and $\bm{g}_k^t$ ($\overline{\bm{g}}_k^t$) as the exact (estimated) gradient for the $k$-th quantum processor at the $t$-th iteration.

\subsection{Some properties of the loss function}

As explained in the main text, the loss function to be minimized in VQE is  
\begin{equation}\label{equ:lf}
    L(\bm{\theta})=\Tr(H\bm{\rho}(\bm{\theta})),
\end{equation}
where $H$ is the problem Hamiltonian and $\bm{\rho}(\bm{\theta})$ is the density matrix of the prepared quantum system parameterized by $\bm{\theta}$. Without loss of generality, the problem Hamiltonian $H$ is expressed as a weighted summation of Pauli operators $H=\sum_{i=1}^M\alpha_iH_i\in \mathbb{C}^{2^n\times 2^n}$, where $H_i\in\{\sigma_X, \sigma_Y, \sigma_Z, \sigma_I\}^{\otimes n}$. The properties of the loss function $L$ are summarized in following four lemmas, which quantify the bounded gradient and Lipschitz continuity and will be employed in the subsequent context.

\begin{lemma}[Bounded gradient norm of the loss function]\label{lemma:bg}  
The norm of the gradient of loss function $L$ with respect to parameter $\bm{\theta}$ is bounded by a constant $\left\|\frac{\partial L(\bm{\theta})}{\partial \bm{\theta}} \right\|\leq G$, where $G=P\sum_{i=1}^M|\alpha_i|$, and $P$ is the dimension of the parameters $\bm{\theta}$.
\end{lemma}

\begin{proof}[Proof of Lemma \ref{lemma:bg}]
    For the $i$-th parameter $\bm{\theta}_i$, we can obtain the exact gradient by the parameter-shift rule 
\begin{equation}\label{eqn:ps}
    \frac{\partial L(\bm{\theta})}{\partial \bm{\theta}_i}=\frac{1}{2}(L(\bm{\theta}+\frac{\pi}{2}\bm{e}_i)-L(\bm{\theta}-\frac{\pi}{2}\bm{e}_i)),
\end{equation}
where $\bm{e}_i$ is an indicator vector for the $i$-th element. Recall that
\begin{equation}\label{eqn:loss_bound}
    L(\bm{\theta})=\Tr(\rho(\bm{\theta})H)=\sum_{i=1}^M\alpha_i\Tr(\rho(\bm{\theta})H_i)\leq \sum_{i=1}^M|\alpha_i|,
\end{equation}
where the inequality holds because $|\Tr(\rho(\bm{\theta})H_i)|\leq 1$ when $H_i\in\{\sigma_X, \sigma_Y, \sigma_Z, \sigma_I\}^{\otimes n}$. The relation $\left\|\frac{\partial L(\bm{\theta})}{\partial \bm{\theta}} \right\|\leq G$ with $G=P\sum_{i=1}^M|\alpha_i|$ can be achieved by substituting Eq.~(\ref{eqn:loss_bound}) into Eq.~(\ref{eqn:ps}).
\end{proof}

\begin{lemma}[$F_1$-Lipschitz continuity of the loss function \cite{sweke2020stochastic}]\label{lemma:F1}
The loss function $L(\bm{\theta})=\Tr(\rho(\bm{\theta})H)$ is $F_1$-Lipschitz continuous $\left|L(\bm{\theta})-L(\bm{\beta})\right|\leq F_1\left\|\bm{\theta}-\bm{\beta}\right\|$ with $F_1=G$.
\end{lemma}

\begin{proof}[Proof of Lemma \ref{lemma:F1}]
    Recall the mean value theorem, for a differentiable loss function $L$, $\exists \bm{\gamma}\in(\bm{\theta}, \bm{\beta})$ such that
\begin{equation}
    L(\bm{\theta})-L(\bm{\beta})=\braket{\frac{\partial L(\bm{\gamma})}{\partial \bm{\gamma}},\bm{\theta}-\bm{\beta}}.
\end{equation}
Furthermore, using $\braket{x,y}\leq \left\|x\right\|\left\|y\right\|$, we have
\begin{equation}
    \begin{aligned}
        \braket{\frac{\partial L(\bm{\gamma})}{\partial \bm{\gamma}},\bm{\theta}-\bm{\beta}} \leq \left\|\frac{\partial L(\bm{\gamma})}{\partial \bm{\gamma}}\right\|\left\|\bm{\theta}-\bm{\beta}\right\|\leq G\left\|\bm{\theta}-\bm{\beta}\right\|,
    \end{aligned}
\end{equation}
where the second inequality holds based on Lemma \ref{lemma:bg}. This lead to $F_1=G$.
\end{proof}

\begin{lemma}[$F_2$-Lipschitz continuity of gradient]\label{lemma:F2}
Define a map $g:[0,2\pi)^P\rightarrow \mathbb{R}^P$ formulated as $g(\bm{\theta})=\frac{\partial L(\bm{\theta})}{\partial \bm{\theta}}$. Then $g(\bm{\theta})$ is $F_2$-Lipschitz continuous $\left\|g(\bm{\theta})-g(\bm{\beta})\right\|\leq F_2\left\|\bm{\theta}-\bm{\beta}\right\|$ with $F_2=PG$.
\end{lemma}

\begin{proof}[Proof of Lemma \ref{lemma:F2}]
    Combining Eq.~\ref{eqn:ps} and Lemma \ref{lemma:F1}, we have
\begin{equation}
    \begin{aligned}
        |g(\bm{\theta}_i)-g(\bm{\beta}_i)|&=\frac{1}{2}|L(\bm{\theta}+\frac{\pi}{2}\bm{e}_i)-L(\bm{\theta}-\frac{\pi}{2}\bm{e}_i)-L(\bm{\beta}+\frac{\pi}{2}\bm{e}_i)+L(\bm{\beta}-\frac{\pi}{2}\bm{e}_i)|\\
        &\leq \frac{1}{2}(|L(\bm{\theta}+\frac{\pi}{2}\bm{e}_i)-L(\bm{\beta}+\frac{\pi}{2}\bm{e}_i)|+|L(\bm{\beta}-\frac{\pi}{2}\bm{e}_i)-L(\bm{\theta}-\frac{\pi}{2}\bm{e}_i)|)\\
        &\leq F_1|\bm{\theta}-\bm{\beta}|,
    \end{aligned}
\end{equation}
where the first inequality follows the triangle inequality $|x+y|\leq|x|+|y|$ and the second inequality is directly derived from Lemma \ref{lemma:F1}. Therefore, for a quantum state controlled by $P$ parameters, we have $\left\|g(\bm{\theta})-g(\bm{\beta})\right\|\leq F_2\left\|\bm{\theta}-\bm{\beta}\right\|$ with $F_2=PF_1$.
\end{proof}

\begin{lemma}[Gradient of the noisy loss function]\label{lemma:dn}
Let $\mathcal{N}_p$ be the global depolarizing channel with the strength $p$. For a quantum state $\rho(\bm{\theta})$ parameterized by $\bm{\theta}$, the gradient after applying the depolarizing channel is \[\frac{\partial L(\mathcal{N}_p(\rho(\bm{\theta})))}{\partial \theta_i}=(1-p)\frac{\partial L(\rho(\bm{\theta}))}{\partial \theta_i}.\]
\end{lemma}

\begin{proof}[Proof of Lemma \ref{lemma:dn}]
    Recalling the depolarizing channel $\mathcal{N}_p(\rho)=(1-p)\rho+p\mathbb{I}/2^n$, we have the noisy loss function
\begin{equation}
    \overline{L}(\bm{\theta})=L(\mathcal{N}_p(\rho(\bm{\theta})))=\Tr(H((1-p)\rho(\bm{\theta})+p\mathbb{I}/2^n))=(1-p)\Tr(H\rho(\bm{\theta}))+p\Tr(H/2^n)=(1-p)L(\bm{\theta})+p\Tr(H/2^n).
\end{equation}
Following the parameter shift rule in Eq.~(\ref{eqn:ps}), the estimated gradient under the depolarizing noise is 
\begin{equation}
    \frac{\partial \overline{L}(\bm{\theta})}{\partial \theta_i}=\frac{1}{2}(\overline{L}(\bm{\theta}+\frac{\pi}{2}\bm{e}_i)-\overline{L}(\bm{\theta}-\frac{\pi}{2}\bm{e}_i))=\frac{1-p}{2}(L(\bm{\theta}+\frac{\pi}{2}\bm{e}_i)-L(\bm{\theta}-\frac{\pi}{2}\bm{e}_i))=(1-p)\frac{\partial L(\bm{\theta})}{\partial \theta_i}.
\end{equation}
\end{proof}

\section{Proof of Lemma \ref{the:shuffle}}\label{pr:shuffle}

For each iteration, $L_m$ is estimated by the randomly sampled $m$ terms without replacement. Therefore, we have
\begin{equation}
    \mathbb{E}_{\pi}[L_m]=\mathbb{E}_{\pi}[\Tr(\rho(\bm{\theta})\sum_{i=1}^mH_{\pi(i)})]=\sum_{i=1}^m\Tr(\rho(\bm{\theta})\mathbb{E}_{\pi}[H_{\pi(i)}])=\frac{1}{M}\sum_{i=1}^m\Tr(\rho(\bm{\theta})\sum_{j=1}^MH_j)=\frac{m}{M}L.
\end{equation}

According to the parameter-shift rule, the exact gradient of $L_m$ with respect to the $i$-th parameter $\bm{\theta}_i$ is expressed as
\begin{equation}\label{eqn:psm}
    \frac{\partial L_m(\bm{\theta})}{\partial \theta_i}=\frac{1}{2}(L_m(\bm{\theta}+\frac{\pi}{2}\bm{e}_i)-L_m(\bm{\theta}-\frac{\pi}{2}\bm{e}_i))=\frac{m}{M}\frac{\partial L(\bm{\theta})}{\partial \theta_i},
\end{equation}
where $\bm{e}_i$ is an indicator vector for the $i$-th element. Without loss of generality, each trainable parameter is assumed to be mutually independent. The expectation of gradient vector is expressed as
\begin{equation}
    \mathbb{E}[\frac{\partial L_m}{\partial \bm{\theta}}]=\frac{m}{M}\frac{\partial L}{\partial \bm{\theta}}.
\end{equation}

\section{Proof of Lemma \ref{lemma:grad-bias}}\label{pr:grad-bias}

With the condition that the norm of local gradient of each processor is bounded by $||\bm g_k(\bm\theta, H_k)||^2\leq G^2$, we have $||\nabla L(\bm\theta,H)||^2=||\sum_{k=1}^K\bm g_k(\bm\theta, H_k)||^2\leq K^2G^2$.

For QUDIO, the discrepancy between local gradient and global gradient is bounded by
\begin{equation}\label{eqn:fix-ideal}
    \begin{aligned}
        ||\nabla L(\bm\theta^t_k,H)-\bm g^t_k(\bm\theta_k^t, H_k)||^2&\leq 2(||\nabla L(\bm\theta^t_k,H)||^2+||\bm g^t_k(\bm\theta_k^t, H_k)||^2)\\
        &=2(K^2+1)G^2,
    \end{aligned}
\end{equation}
where the first inequality follows $(a-b)^2\leq 2(a^2+b^2)$.

For Shuffle-QUDIO, the discrepancy $\mathbb{E}_{H_k|\bm\theta_k^t}[||\nabla L(\bm\theta^t_k,H)-\bm g^t_k(\bm\theta_k^t, H_k)||^2]$ is quantified by taking expectation over the randomly shuffling Hamitonians $H_k$ given parameters $\bm{\theta}_k^t$, i.e., 

\begin{equation}\label{eqn:shuffle-ideal}
    \begin{aligned}
        &\mathbb{E}_{H_k|\bm\theta_k^t}[||\nabla L(\bm\theta^t_k,H)-\bm g^t_k(\bm\theta_k^t, H_k)||^2]\\
    =&\mathbb{E}_{H_k|\bm\theta_k^t}[||\nabla L(\bm\theta^t_k,H)||^2]-2\mathbb{E}_{H_k|\bm\theta_k^t}[\nabla^T L(\bm\theta^t_k,H)\bm g^t_k(\bm\theta_k^t, H_k)]\\
        &+\mathbb{E}_{H_k|\bm\theta_k^t}[||\bm g^t_k(\bm\theta_k^t, H_k)||^2]\\
        =&(1-\frac{2}{K})\mathbb{E}_{H_k|\bm\theta_k^t}[||\nabla L(\bm\theta^t_k,H)||^2]+\mathbb{E}_{H_k|\bm\theta_k^t}[||\bm g^t_k(\bm\theta_k^t, H_k)||^2]\\
        \leq& (1-\frac{2}{K})||\nabla L(\bm\theta^t_k,H)||^2+G^2\\
        \leq&(K-1)^2G^2,
    \end{aligned}
\end{equation}
where the second equality is derived by utilizing Lemma \ref{the:shuffle}, the first inequality comes from the bound of gradient norm in Lemma 3, and the last inequality uses the induced bound of $||\nabla L(\bm\theta,H)||^2$ and holds when $K\geq 2$ required by the condition $1-\frac{2}{K}\geq 0$.

\section{Proof of Theorem \ref{the:convergence}}\label{app:proof1}

Introduce the ancillary variables
\begin{gather}
    \bm\theta^{t} = \frac{1}{K}\sum_{k=1}^K\bm\theta^t_k,\quad \overline{\bm g}^t=\frac{1}{K}\sum_{k=1}^K\overline{\bm g}^t_k, \quad\bm\theta^{t+1}-\bm\theta^t=\eta \frac{1}{K}\sum_{k=1}^K\overline{\bm g}^t_k=\eta \overline{\bm g}^t.
\end{gather}

According to the $F$-Lipschitz continuity of loss function in Lemma \ref{lemma:F1}, we have
\begin{equation}\label{eqn:lipschitz}
    \begin{aligned}
        L(\bm\theta^{t+1})-L(\bm\theta^t)&\leq -\eta\braket{\nabla L(\bm\theta^t),\frac{1}{K}\sum_{k=1}^K\overline{\bm g}^t_k}+\frac{\eta^2F}{2K^2}||\sum_{k=1}^K\overline{\bm g}^t_k||^2\\
        &=-\frac{\eta}{2}||\nabla L(\bm\theta^t)||^2-\frac{\eta}{2K^2}||\sum_{k=1}^K\overline{\bm g}^t_k||^2+\frac{\eta}{2}||\nabla L(\bm\theta^t)-\frac{1}{K}\sum_{k=1}^K\overline{\bm g}^t_k||^2+\frac{\eta^2F}{2K^2}||\sum_{k=1}^K\overline{\bm g}^t_k||^2\\
        &=-\frac{\eta}{2}||\nabla L(\bm\theta^t)||^2+\frac{\eta}{2}\underbrace{||\nabla L(\bm\theta^t)-\frac{1}{K}\sum_{k=1}^K\overline{\bm g}^t_k||^2}_{T1}+\frac{\eta(\eta F-1)}{2K^2}||\sum_{k=1}^K\overline{\bm g}^t_k||^2,
    \end{aligned}
\end{equation}
where $\nabla L(\bm\theta^t)=\nabla L(\bm\theta^t, H)=\sum_{k=1}^K\bm g^t_k(\bm\theta^t, H_k^t)$ and $\bm g^t_k=\bm g^t_k(\bm\theta^t_k, H_k^t)$. Note that $H=\sum_{k=1}^MH_k$ is a constant all the time, which is the reason why the superscript $t$ is discarded. The first equality is derived by utilizing $\braket{\bm a,\bm b}=\frac{1}{2}(||\bm a||^2+||\bm b||^2-||\bm a-\bm b||^2)$. The second equality holds by merging the second term and the fourth term in the first equality.

Next,  the term $T1$ in Eq.~(\ref{eqn:lipschitz}) yields

\begin{equation}\label{eqn:grad_diff}
    \begin{aligned}
        ||\nabla L(\bm\theta^t)-\frac{1}{K}\sum_{k=1}^K\overline{\bm g}^t_k||^2&=||\nabla L(\bm\theta^t,H)-\frac{1}{K}\sum_{k=1}^K\overline{\bm g}^t_k(\bm\theta_k^t, H_k)||^2\\
        &=||\frac{1}{K}\sum_{k=1}^K\nabla L(\bm\theta^t,H)-\frac{1}{K}\sum_{k=1}^K\overline{\bm g}^t_k(\bm\theta_k^t, H_k)||^2\\
        &=||\frac{1}{K}\sum_{k=1}^K(\nabla L(\bm\theta^t,H)-\overline{\bm g}^t_k(\bm\theta_k^t, H_k))||^2\\
        &\leq\frac{1}{K}\sum_{k=1}^K||\nabla L(\bm\theta^t,H)-\overline{\bm g}^t_k(\bm\theta_k^t, H_k)||^2\\
        &=\frac{1}{K}\sum_{k=1}^K||\nabla L(\bm\theta^t,H)-\nabla L(\bm\theta^t_k,H)+\nabla L(\bm\theta^t_k,H)-\overline{\bm g}^t_k(\bm\theta_k^t, H_k)||^2\\
        &\leq\frac{2}{K}\sum_{k=1}^K[||\nabla L(\bm\theta^t,H)-\nabla L(\bm\theta^t_k,H)||^2+||\nabla L(\bm\theta^t_k,H)-\overline{\bm g}^t_k(\bm\theta_k^t, H_k)||^2]\\
        &\leq\frac{2}{K}\sum_{k=1}^K[F^2||\bm\theta^t-\bm\theta^t_k||^2+\underbrace{||\nabla L(\bm\theta^t_k,H)-\overline{\bm g}^t_k(\bm\theta_k^t, H_k)||^2}_{T2}],
    \end{aligned}
\end{equation}
where the first inequality follows the Jensen's inequality $||\frac{1}{n}\sum_{i=1}^n\bm a_i||^2\leq \frac{1}{n}\sum_{i=1}^n||\bm a_i||^2$, the second inequality holds because of the triangle inequality $||\bm a+\bm b||\leq||\bm a||+||\bm b||$, the last inequality is derived by $F$-Lipschitz continuity condition of gradient function in Lemma  \ref{lemma:F2}. 

We first calculate the upper bound of $||\bm\theta^t-\bm\theta^t_k||^2$. Assume the latest model synchronization happens at the iteration $t_c$ with  $t-t_c<W$, then $\bm\theta_k^{t_c+1}=\bm\theta^{t_c}$. According to the gradient descent rule, the parameter $\bm\theta_k^t$ is derived as
\begin{equation}
    \bm\theta^t_k=\bm\theta^{t-1}_k-\eta \overline{\bm g}_k^{t-1}(\bm\theta^{t-1}_k,H_k)=\bm\theta^{t_c+1}-\sum_{j=t_c+1}^{t-1}\eta \overline{\bm g}_k^j(\bm\theta^j_k,H_k),
\end{equation}
where the subscript of parameters $\bm\theta^{t_c+1}$ in the second equality is discarded without ambiguity because $\bm\theta^{t_c+1}_k=\bm\theta^{t_c+1}_l,\forall k,l\in \{1,...,K\}$. Then
\begin{equation}
    \bm\theta^t=\frac{1}{K}\sum_{k=1}^K\bm\theta^t_k=\bm\theta^{t_c+1}-\frac{1}{K}\sum_{k=1}^K\sum_{j=t_c+1}^{t-1}\eta \overline{\bm g}_k^j(\bm\theta^j_k,H_k)
\end{equation}

Therefore, the deviation between local quantum model $\bm\theta^t_l$ (note that we use notation $\bm\theta^t_l$ instead of $\bm\theta^t_k$ to avoid the confusion between a specified quantum worker $l$ and the general representation of the $k$-th worker    in the following derivation) and the  average model $\bm\theta^t$ is measured as
\begin{equation}\label{eqn:theta_diff}
    \begin{aligned}
        ||\bm\theta^t_l-\bm\theta^t||^2 &= ||\frac{1}{K}\sum_{k=1}^K\sum_{j=t_c+1}^{t-1}\eta \overline{\bm g}_k^j(\bm\theta^j_k,H_k)-\sum_{j=t_c+1}^{t-1}\eta \overline{\bm g}_l^j(\bm\theta^j_l,H_l)||^2\\
        &=\eta^2||\frac{1}{K}\sum_{k=1}^K\sum_{j=t_c+1}^{t-1}(\overline{\bm g}_k^j(\bm\theta^j_k,H_k)- \overline{\bm g}_l^j(\bm\theta^j_l,H_l))||^2\\
        &\leq \eta^2W\sum_{j=t_c+1}^{t-1}||\frac{1}{K}\sum_{k=1}^K(\overline{\bm g}_k^j-\overline{\bm g}_l^j)||^2\\
        &\leq \frac{\eta^2W}{K}\sum_{j=t_c+1}^{t-1}\sum_{k=1}^K||\overline{\bm g}_k^j-\overline{\bm g}_l^j||^2\\
        &\leq \frac{2\eta^2W^2(K-1)(1-p)^2G^2}{K},
    \end{aligned}
\end{equation}
where the first and second inequalities follows $||\sum_{i=1}^Na_i||^2\leq N\sum_{i=1}^N||a_i||^2$, and the last inequality is deduced based on the bound of gradient.

We now derive the upper bound of the second term $T2$ in the last inequality. Recall Lemma \ref{lemma:bg} such that the norm of the gradient of each worker is bounded by $||\bm g_k(\bm\theta, H_k)||^2\leq G^2$. We have $||\nabla L(\bm\theta,H)||^2=||\sum_{k=1}^K\bm g_k(\bm\theta, H_k)||^2\leq K^2G^2$ and $||\overline{\bm g}_k(\bm\theta, H_k)||^2\leq (1-p)^2G^2$.

\noindent\textit{\underline{Convergence of QUDIO.}} For the case of QUDIO where each quantum processor is assigned fixed partial Hamiltonian terms at the beginning, $T2$ is bounded by
\begin{equation}\label{eqn:fix}
    \begin{aligned}
        ||\nabla L(\bm\theta^t_k,H)-\overline{\bm g}^t_k(\bm\theta_k^t, H_k)||^2&\leq 2(||\nabla L(\bm\theta^t_k,H)||^2+||\overline{\bm g}^t_k(\bm\theta_k^t, H_k)||^2)\\
        &=2(K^2+(1-p)^2)G^2.
    \end{aligned}
\end{equation}

Combining Eqs. (\ref{eqn:grad_diff}),  (\ref{eqn:fix})  and (\ref{eqn:theta_diff}), we can quantify the progress of one global iteration under the fixed Hamiltonian partition strategy as
\begin{equation}
    \begin{aligned}
        L(\bm\theta^{t+1})-L(\bm\theta^t)&\leq -\frac{\eta}{2}||\nabla L(\bm\theta^t)||^2+\frac{\eta}{2}||\nabla L(\bm\theta^t)-\frac{1}{K}\sum_{k=1}^K\overline{\bm g}^t_k||^2+\frac{\eta(\eta F-1)}{2K^2}||\sum_{k=1}^K\overline{\bm g}^t_k||^2\\
        &\leq -\frac{\eta}{2}||\nabla L(\bm\theta^t)||^2+\frac{\eta}{2}\frac{2}{K}\sum_{k=1}^K[F^2||\bm\theta^t-\bm\theta^t_k||^2+||\nabla L(\bm\theta^t_k,H)-\overline{\bm g}^t_k(\bm\theta_k^t, H_k)||^2]+\frac{\eta(\eta F-1)}{2K^2}||\sum_{k=1}^K\overline{\bm g}^t_k||^2\\
        &\leq -\frac{\eta}{2}||\nabla L(\bm\theta^t)||^2+\frac{\eta}{2}\frac{2}{K}\sum_{k=1}^K[F^2\frac{2\eta^2W^2(K-1)(1-p)^2G^2}{K}+2(K^2+(1-p)^2)G^2]+\frac{\eta(\eta F-1)(1-p)^2G^2}{2}\\
        &= -\frac{\eta}{2}||\nabla L(\bm\theta^t)||^2+F^2\frac{2\eta^3W^2(K-1)(1-p)^2G^2}{K}+2\eta(K^2+(1-p)^2)G^2+\frac{\eta(\eta F-1)(1-p)^2G^2}{2}
    \end{aligned}
\end{equation}
where the last term in the last inequality follows $||\sum_{i=1}^n\bm a_i||^2\leq n\sum_{i=1}^n||\bm a_i||^2$.

Rearranging the inequality above and summing over $t$, we achieve
\begin{equation}\label{eqn:convergence-fix}
    \begin{aligned}
        \frac{1}{T}\sum_{t=1}^T||\nabla L(\bm\theta^t)||^2&\leq \frac{2(L(\bm\theta^1)-L(\bm\theta^{T+1}))}{\eta T}+F^2\frac{4\eta^2W^2(K-1)(1-p)^2G^2}{KT}+\frac{4(K^2+(1-p)^2)G^2}{T}+\frac{(\eta F-1)(1-p)^2G^2}{T}\\
        &= \frac{2(L(\bm\theta^1)-L(\bm\theta^{T+1}))}{\eta T}+\frac{(4F^2\eta^2W^2(K-1)(1-p)^2+\bm{4K^3}+K(\eta F+3)(1-p)^2)G^2}{KT}
    \end{aligned}
\end{equation}

\medskip
\noindent\textit{\underline{Convergence of Shuffle-QUDIO.}}
For the case of Shuffle-QUDIO where the whole Hamiltonian terms are shuffled and reassigned to local processors during each iteration, we can obtain a tighter bound for $T2$ by taking expectation over the random local Hamiltonians $H_k$ given parameters $\bm\theta_k^t$
\begin{equation}\label{eqn:shuffle}
    \begin{aligned}
        \mathbb{E}_{H_k|\bm\theta_k^t}[||\nabla L(\bm\theta^t_k,H)-\overline{\bm g}^t_k(\bm\theta_k^t, H_k)||^2]=&\mathbb{E}_{H_k|\bm\theta_k^t}[||\nabla L(\bm\theta^t_k,H)||^2]-2\mathbb{E}_{H_k|\bm\theta_k^t}[\nabla^T L(\bm\theta^t_k,H)\overline{\bm g}^t_k(\bm\theta_k^t, H_k)]\\
        &+\mathbb{E}_{H_k|\bm\theta_k^t}[||\overline{\bm g}^t_k(\bm\theta_k^t, H_k)||^2]\\
        \leq& ||\nabla L(\bm\theta^t_k,H)||^2-\frac{2(1-p)}{K}||\nabla L(\bm\theta^t_k,H)||^2+(1-p)^2G^2\\
        \leq&(K-1+p)^2G^2,
    \end{aligned}
\end{equation}
where the first inequality is based on Lemma \ref{the:shuffle} and property of bounded gradients of each worker, and the second inequality is derived by applying the induced bound of $||\nabla L(\bm\theta^t_k,H)||^2$. Note that we introduce the implicit constraint $K>2(1-p)$ to assure the coefficient of $||\nabla L(\bm\theta^t_k,H)||^2$ in the first inequality is non-negative.

On the other hand, when shuffling the Hamiltonian terms in every iteration, the loss reduction is bounded by
\begin{equation}
    \begin{aligned}
        \mathbb{E}_{\bm{h}|\bm\theta^t}[L(\bm\theta^{t+1})-L(\bm\theta^t)]\leq& -\frac{\eta}{2}||\nabla L(\bm\theta^t)||^2+\frac{\eta}{2}\mathbb{E}_{\bm{h}|\bm\theta^t}[||\nabla L(\bm\theta^t)-\frac{1}{K}\sum_{k=1}^K\overline{\bm g}^t_k||^2]+\frac{\eta(\eta F-1)}{2K^2}\mathbb{E}_{\bm{h}|\bm\theta^t}[||\sum_{k=1}^K\overline{\bm g}^t_k||^2]\\
        \leq& -\frac{\eta}{2}||\nabla L(\bm\theta^t)||^2+\frac{\eta}{2}\frac{2}{K}\sum_{k=1}^K[F^2||\bm\theta^t-\bm\theta^t_k||^2+\mathbb{E}_{\bm{h}|\bm\theta^t}[||\nabla L(\bm\theta^t_k,H)-\overline{\bm g}^t_k(\bm\theta_k^t, H_k)||^2]]\\
        &+\frac{\eta(\eta F-1)}{2K^2}\mathbb{E}_{\bm{h}|\bm\theta^t}[||\sum_{k=1}^K\overline{\bm g}^t_k||^2]\\
        =&-\frac{\eta}{2}||\nabla L(\bm\theta^t)||^2+\frac{\eta}{2}\frac{2}{K}\sum_{k=1}^K[F^2\frac{2\eta^2W^2(K-1)(1-p)^2G^2}{K}+(K-1+p)^2G^2]\\
        &+\frac{\eta(\eta F-1)(1-p)^2G^2}{2}\\
        =&-\frac{\eta}{2}||\nabla L(\bm\theta^t)||^2+F^2\frac{2\eta^3W^2(K-1)G^2}{K}+\eta(K-1+p)^2G^2+\frac{\eta(\eta F-1)(1-p)^2G^2}{2}.
    \end{aligned}
\end{equation}

Then the convergence of gradient is quantified by
\begin{equation}\label{eqn:convergence-shuffle}
    \begin{aligned}
        \frac{1}{T}\sum_{t=1}^T||\nabla L(\bm\theta^t)||^2&\leq \frac{2(L(\bm\theta^1)-L(\bm\theta^{T+1}))}{\eta T}+\frac{(4F^2\eta^2W^2(K-1)+\bm{2K^2(K-2+2p)}+K(\eta F+1)(1-p)^2)G^2}{KT}.
    \end{aligned}
\end{equation}

Comparing Eqs. (\ref{eqn:convergence-fix}) and (\ref{eqn:convergence-shuffle}), especially for the terms highlighted by the bold face, it is obvious that Shuffle-QUDIO achieves faster convergence than original QUDIO.

\section{Aggregation methods for quantum distributed optimization}\label{app:aggre}

The framework of QUDIO can be roughly deconstructed into two alternating operations, including \textit{local updates} and \textit{global synchronization}. While the former operation is upgraded by introducing the shuffle operation, the latter can be also modified by more advanced techniques. In the original Shuffle-QUDIO in Alg.~\ref{alg:shuffle-qudio}, the average aggregation method is simply employed to merge the information from distributed nodes, which may be sub-optimal without considering the differences among these nodes. In this section, we give detailed descriptions about another three novel aggregation algorithms.

\textbf{Random aggregation.} For each distributed quantum model with parameter $\bm{\theta}_{i}^{(t,W)}$ after completing local updates, the synchronized parameter $\bm{\theta}^{t+1}=\bm{\theta}_{j}^{(t,W)}$, where $j\in [K]$ is randomly generated.

\textbf{Median aggregation.} Different from average aggregation and random aggregation, median aggregation utilizes the loss value of every local quantum model as reference to determine the synchronized quantum model. To be concrete, assuming the loss of the $i$-th quantum processor at the $t$-th synchronization is denoted by $L^{(t)}_i$, the quantum model run on processor $j$ whose loss value $L^{(t)}_j$ is the median of $\{L^{(t)}_i\}_{i=1}^K$ is selected as the synchronized model $\bm{\theta}^{t+1}=\bm{\theta}_{j}^{(t,W)}$.

\textbf{Weighted aggregation.} Instead of simply utilizing a single local quantum model with median loss as the aggregated model, weighted aggregation merges all local quantum models in a weighted summation fashion. The direct motivation is that a local quantum model achieving lower loss should contribute more to the synchronized model. Similar to median aggregation method, we first collect the loss value of every local quantum model $\{L^{(t)}_i\}_{i=1}^K$ to compute the weights. Specifically, a monotone decreasing function is first applied to the loss and then a softmax function is adopted to obtained a normalized weight vector. The mathematical process is formulated as
\begin{equation}
    w^{(t)}_i=\frac{\exp{(-L^{(t)}_i)}}{\sum_{j\in [K]}\exp{(-L^{(t)}_j)}}.
\end{equation}
Based on the loss-induced weights, the synchronized quantum model is obtained as $\bm{\theta}^{(t + 1)} = \sum_{i=1}^K w^{(t)}_i\bm{\theta}^{(t,W)}_i$.

\section{Optimization trajectory visualization of Shuffle-QUDIO}\label{app:traj}

\begin{figure}[htp]
\centering
\includegraphics[width=0.9\textwidth]{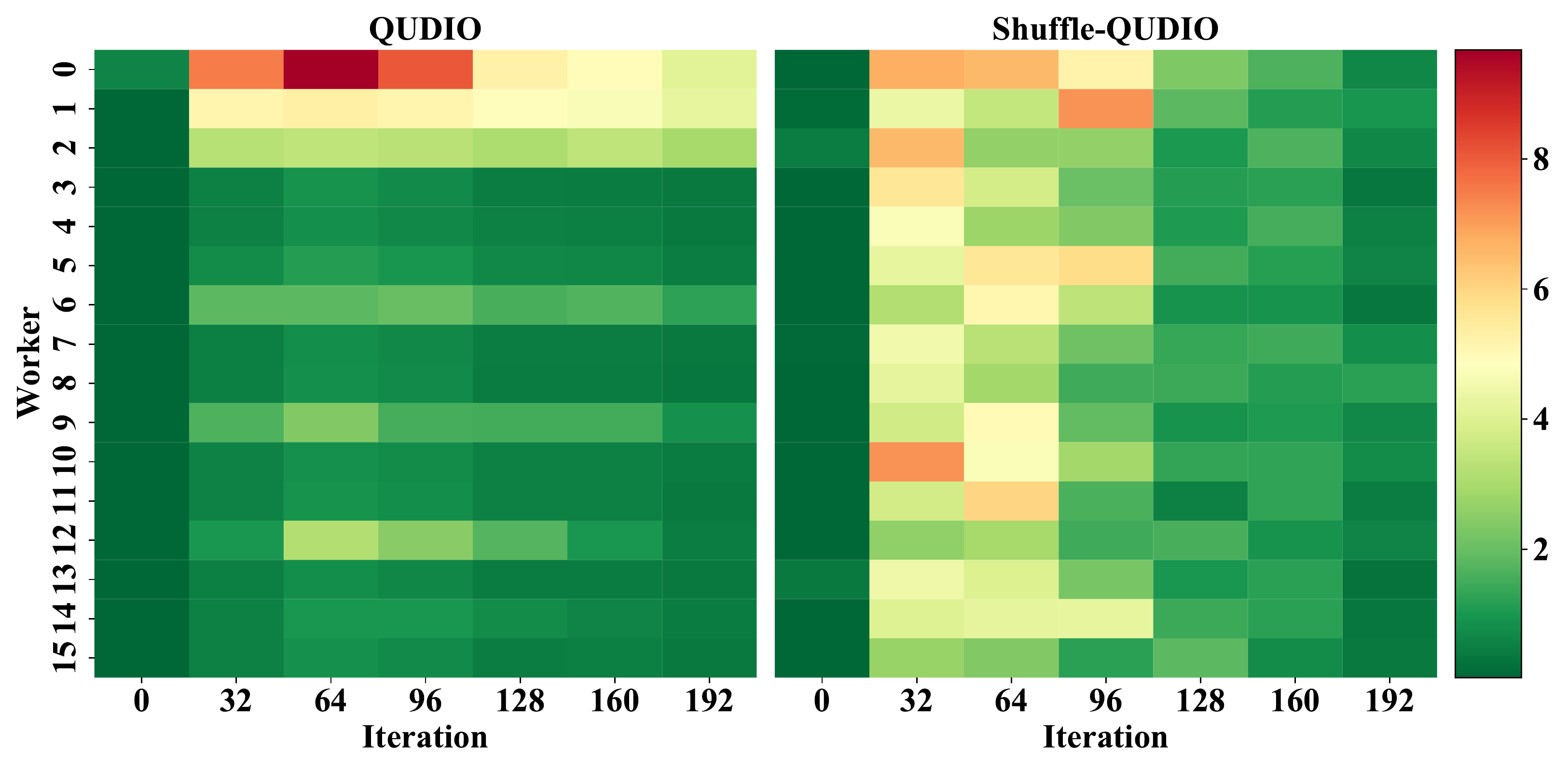}
\caption{\small{\textbf{Convergence of local quantum models to the global synchronized model.} The longitudinal axis represents the sequence number of distributed workers. The color of each grid represents the distance to the synchronized model.}}
\label{fig:param_traj}
\end{figure}

To investigate the potential changes brought by the shuffle operation that leads to faster convergence of Shuffle-QUDIO, we analyze the parameter trajectory of each worker during the whole training process. To be more concrete, the parameters $\bm{\theta}_k^t$ of the quantum circuits run on each local processor are collected before every global synchronization. Then we measure the distance between the local parameters and synchronized parameters after the average aggregation $\left\|\bm{\theta}_k^t-\frac{1}{K}\sum_{k=1}^K\bm{\theta}_k^t\right\|$. 

As shown in Fig.~\ref{fig:param_traj}, there are significant differences about parameter trajectory between vanilla QUDIO and Shuffle-QUDIO. In the early stage of optimization (the number of iterations is less than $128$), almost every local quantum model in Shuffle-QUDIO witnesses a large distance to the synchronized model. This phenomenon reveals the diversity of local models, indicating that each local model can effectively capture the nature of different aspects of the whole system. As the training goes on, the distance gradually decreases and each local model converges to the same point. By contrast, the local models in vanilla QUDIO suffer from severe bias and insufficient training. It appears that the synchronized model is always dominated by some local models, resulting in that other local models do not get trained sufficiently and the final synchronized model fails to work well for the complete system.

\end{document}